\documentclass[english,11pt]{article}
\usepackage{microtype}
\usepackage{graphicx}
\usepackage[scriptsize]{subfigure}
\usepackage{booktabs}

\usepackage{geometry}
\usepackage{amssymb}
\usepackage{amsmath}
\usepackage{amsthm}
\usepackage{xspace}

\usepackage{caption}
\usepackage{algorithm}
\usepackage{algorithmic}

\makeatletter

\numberwithin{equation}{section}
\numberwithin{figure}{section}
\theoremstyle{plain}
	\newtheorem{theorem}{Theorem}[section]

	\newtheorem{lem}[theorem]{Lemma}
	\newtheorem{lemma}[theorem]{Lemma}

\theoremstyle{definition}
	
	\newtheorem{defn}[theorem]{Definition}
	\newtheorem*{remark*}{Remark}
 
\pdfoutput=1

\newif\ifusebb
\usebbtrue
\usebbfalse
\ifusebb
  \usepackage{bbm}
  \DeclareSymbolFont{bbold}{U}{bbold}{m}{n}
  \DeclareSymbolFontAlphabet{\mathbbold}{bbold}

\else

\fi
 
\renewcommand{\epsilon}{\varepsilon}

\usepackage{enumitem}
\setlist[description,1]{align=left,leftmargin=0.1in} 
\setlist[itemize,1]{leftmargin=*}

\usepackage{tikz,wrapfig}
\usepackage{subfigure}

\usepackage[lined,boxed,ruled,norelsize,algo2e]{algorithm2e}

\newenvironment{fminipage}%
  {\begin{Sbox}\begin{minipage}}%
  {\end{minipage}\end{Sbox}\fbox{\TheSbox}}

\def\pvec#1#2{\vec{\mbox{P}}^{#1}\left[ #2 \right]}

\def\otil{\widetilde{O}}

\renewcommand\AA{\mvar{A}}

\newcommand\DD{\mvar{D}}

\newcommand\LL{\mvar{L}}

\global\long\def\R{\mathbb{{R}}}

\global\long\def\mvar#1{\boldsymbol{\mathit{#1}}}

\global\long\def\vvar#1{\boldsymbol{\mathit{#1}}}

\newcommand{\vb}{\vvar{b}}
\newcommand{\vc}{\vvar{c}}
\newcommand{\vd}{\vvar{d}}
\newcommand{\ve}{\vvar{e}}

\newcommand{\vp}{\vvar{p}}

\newcommand{\vr}{\vvar{r}}
\newcommand{\vs}{\vvar{s}}

\newcommand{\vx}{\vvar{x}}
\newcommand{\vy}{\vvar{y}}

\newcommand{\vphi}{\vvar{\phi}}
\newcommand{\vPhi}{\vvar{\Phi}}
\newcommand{\valpha}{\vvar{\alpha}}

\renewcommand{\pvec}[1]{\vvar{#1}\mkern2mu\vphantom{#1}}
\newcommand{\vrp}{\pvec{r}'}
\newcommand{\vcp}{\pvec{c}'}

\newcommand{\vones}{\ensuremath{\mathbf{1}}}
\newcommand{\vzero}{\ensuremath{\mathbf{0}}}

\newcommand{\vindic}[1]{\ve_{i}}

\newcommand{\mI}{\mvar I}

\newcommand{\eps}{\epsilon}

\newcommand{\mdiag}{\mathbf{D}}

\newcommand\REFlinewhile{5}
\newcommand\REFlineretone{11}
\newcommand\REFlinerettwo{15}
\newcommand\REFlineretthree{20}

\newcommand\REFlinelinfaverageornot{7}
\newcommand\REFlinelinfpertrule{13}

\newcommand\REFlineloneretone{11}
\newcommand\REFlinelonerettwo{15}
\newcommand\REFlineloneretthree{20}

\newcommand\REFlineloneaverageornot{7}
\newcommand\REFlinelonepertrule{13}

\begin{document}

\title{Improved Convergence for $\ell_\infty$ and $\ell_1$ Regression via Iteratively Reweighted Least Squares}

\author{
Alina Ene
\thanks{Department of Computer Science, Boston University, \texttt{aene@bu.edu}.} 
\and
Adrian Vladu 
\thanks{Department of Computer Science, Boston University, \texttt{avladu@mit.edu}.}
}

\date{}
\maketitle
\thispagestyle{empty}
\setcounter{page}{0}

\begin{abstract}
	\normalsize
The iteratively reweighted least squares method (IRLS) is a popular technique used in practice for solving regression problems. Various versions of this method have been proposed, but their theoretical analyses failed to capture the good practical performance. 

In this paper we propose a simple and natural version of IRLS for solving $\ell_\infty$ and $\ell_1$ regression, which provably converges to a $(1+\epsilon)$-approximate solution in $O(m^{1/3}\log(1/\epsilon)/\epsilon^{2/3} + \log m/\epsilon^2)$ iterations, where $m$ is the number of rows of the input matrix. Interestingly, this running time is independent of the conditioning of the input, and the dominant term of the running time depends sublinearly in $\epsilon^{-1}$, which is atypical for the optimization of non-smooth functions.

This improves upon the more complex algorithms of Chin et al. (ITCS '12), and Christiano et al. (STOC '11) by a factor of at least $1/\epsilon^2$, and yields a truly efficient natural algorithm for the slime mold dynamics (Straszak-Vishnoi, SODA '16, ITCS '16, ITCS '17).
\end{abstract}


\section{Introduction}

Regression problems are fundamental primitives in scientific computing. Among these, $\ell_\infty$- and $\ell_1$-regression are their hardest instantiations, since through standard reductions they can be shown to be equivalent to linear programming. 

While the series of works on these topics is truly extensive and diverse, 
the simpler methods have pervaded into the realm of practical applications. Among these, an extremely popular scheme known for its simplicity is the iteratively re-weighted least squares (IRLS) method. The idea behind it is to reduce optimization problems to iteratively solving a series of weighted $\ell_2$-minimization problems, where the weights are adaptively chosen in such a way that the resulting solutions from the sequence of least-squares problems converge to the sought optimal point. In particular, due to its relevance in signal processing, $\ell_1$ regression is a very important application of IRLS~\cite{candes2006compressive,chartrand2008iterative}.

Despite the fact that various versions of this method have been studied ever since the 60's~\cite{lawson1961contributions,osborne1985finite} theoretical understanding  of their convergence has lacked. Recent works have attempted to fill this gap, and offer provable guarantees~\cite{daubechies2010iteratively,StraszakV16,StraszakV16Soda,StraszakV16arxiv}, some of them inspired from the interpretation of this method as a dynamical system. In particular, we note the \textit{Physarum} dynamics, which have been studied in a completely different context~\cite{ito2011convergence,johannson2012slime,tero2007mathematical,bonifaci2012physarum,becchetti2013physarum} in order to justify an experiment which revealed that a unicellular organism, the slime mold, could solve the shortest path problem in a maze~\cite{nakagaki2000intelligence}. The fact that these dynamics are essentially just a version of the IRLS method was observed in~\cite{StraszakV16}.

Returning to the more classical world of algorithm design and analysis, it is worth observing that existing analyses of IRLS methods fall into one of the following two categories: (i) they show convergence only when the problem is properly initialized, or (ii) the guaranteed running time is prohibitive in the sense that it is highly dependent on how the input is conditioned, or it has a high polynomial dependency on the desired solution accuracy.

In this paper, we focus on analyzing simple versions of IRLS which overcome both aforementioned obstacles. In particular, our methods always converge to $1+\epsilon$ multiplicative approximation for the objectives $\min_{\vx:\AA\vx = \vb}\|\vx\|_p$, $p\in\{1,\infty\}$, in $\otil(m^{1/3}/\epsilon^{2/3} + 1/\epsilon^2)$ iterations\footnote{We use $\otil$ notation to suppress polylogarithmic factors in $m/\epsilon$.}  of solving a weighted least squares problem, where $m$ is the dimension of the sought vector $\vx$.

Inspiration for our methods is heavily drawn from the work of~\cite{ChristianoKMST11}, which offered a ground-breaking result by showing that in undirected graphs, a $(1+\epsilon)$-approximate maximum flow can be found in $\otil(m^{1/3}/\epsilon^{11/3})$ iterations (subsequently the $\epsilon$ dependence was improved to $1/\epsilon^{8/3}$, see~\cite{ChinMMP13}) of solving a weighted least squares problem -- which in conjunction with efficient Laplacian system solvers, broke a longstanding barrier for fast graph algorithms. While these algorithms generalize to arbitrary $\ell_1$ and $\ell_\infty$ regression problems, they are  somewhat involved, in particular due to the fact that they are the product of combining the multiplicative weights update method with a regularization technique, and a second potential function.\footnote{To be more specific, Christiano et al. solve the approximate maximum flow problem, which is a specific instance of $\ell_\infty$ regression. Chin et al. build on this work to solve $\ell_1$ regression with block structure; the block structure is relevant for their specific applications, but is a direct extension of the method, so solving vanilla $\ell_1$ regression is still the main problem tackled there.}

Instead, our method attempts to directly solve the non-smooth objective while tracking a single potential function. The number of iterations looks surprising, since the dominant term is $\otil(m^{1/3}/\epsilon^{2/3})$, whenever $\epsilon \geq m^{-1/4}$, while classical techniques for optimizing non-smooth functions require  a number of iterations that depends on the product between the function's parameters (such as Lipschitz constant of the gradient or radius of the domain), and $1/\epsilon$ in the best case, when accelerated methods are used; see~\cite{nesterov2005smooth} for more details.

\footnote{We emphasize that using off-the-shelf methods, without further assumptions on the input, the number of iterations of any standard optimization method would be $\Omega(\sqrt{m})$ even for the very special instances where the affine constraint corresponds to a flow satisfying a given demand in unweighted graphs, and in general will depend on how the input matrix is conditioned, since this conditioning determines the magnitude of the subgradients or the diameter of the domain we are optimizing over. The breakthrough of Christiano et al. was the first work that managed to reduce this dependence for maximum flow, which is a specific instance of the $\ell_\infty$ regression problem.
}

Interestingly, a line of work that yielded results very similar in spirit to ours is that of approximately solving positive linear and semidefinite programs~\cite{young2001sequential,allen2015using,allen2016using}, where the goal was to produce a first order optimization method that can be implemented in a number of iterations independent of the conditioning of the input. Improving the $\epsilon$ dependence to $o(1/\epsilon^2)$ is an important open problem in this subfield.

We believe that understanding the connection between these results can pave the way for designing new efficient optimization primitives. 

\subsection{Main Theorem}
We state the main theorem of this paper. It follows from the correctness proofs described in Sections~\ref{subsec:linf_min} and~\ref{subsec:lone_min}, and the convergence proofs from Lemmas~\ref{lem:linf_conv} and~\ref{lem:lone_conv_fast}.

\begin{theorem}\label{thm:main_theorem}
There exist algorithms \textsc{$\ell_\infty$-Minimization} and \textsc{$\ell_1$-Minimization} such that, on input $(\AA, \vb, \epsilon, M)$, where
$\AA \in \R^{n \times m}$ is a matrix, $\vb \in \R^n$ is a vector which lies in the span of $\AA$'s columns, $\epsilon$ is an accuracy parameter, and $M$ is a target value:
\begin{enumerate}
\item \textsc{$\ell_\infty$-Minimization} returns a solution $\vx$ such that $\AA \vx = \vb$, and $\|\vx\|_\infty \leq (1+\epsilon) M$, or certifies that  $\min_{\vx: \AA \vx = \vb} \|\vx\|_\infty \geq (1-\epsilon) M$.
\item \textsc{$\ell_1$-Minimization} returns a solution $\vx$ such that $\AA \vx = \vb$, and $\|\vx\|_1 \leq (1+\epsilon) M$, or certifies that $\min_{\vx: \AA \vx = \vb} \|\vx\|_1 \geq (1-\epsilon)M$.
\end{enumerate}
Furthermore both algorithms finish in
\[
O\left(\frac{m^{1/3} \log(1/\epsilon)}{\epsilon^{2/3}} + \frac{\log m}{\epsilon^2}\right)
\]
iterations, each of which can be implemented in the time required to solve
a linear system of the form $\AA \DD \AA^\top \vphi = \vb$, where $\DD \in \R^{m\times m}$ is an arbitrary nonnegative diagonal matrix.
\end{theorem}
While our theorem statements are concerned with approximately solving a decision problem which requires a guess $M$ on the value of the objective, it follows from standard techniques that this can be used to find a good approximation to the optimal solution without paying more than a $\otil(1)$ overhead in the number of iterations. For completeness, we provide the details in Section~\ref{sec:appx_opt}.

\subsection{Relation to Previous IRLS Methods and Slime-Mold Dynamics}
A popular method for solving $\ell_1$ minimization is the iteratively re-weighted least squares method (IRLS). This is essentially based on the observation that whenever $\vx^* = \arg\min_{\AA\vx = \vb} \|\vx\|_1$, one also has that this is the minimizer of the least squares problem $\arg\min_{\vx: \AA \vx = \vb} \langle {1}/{\vx^*}, \vx^2\rangle$.\footnote{Throughout the paper we use the convention that $0/0 = 0$.} Hence one approach that has been employed ever since the 60's~\cite{lawson1961contributions, osborne1985finite, daubechies2010iteratively} is to iteratively adjust the weighting of the coordinates and re-solve the least squares problem, until $\vx$ converges to a stationary point. 
This is rigorously described by the iteration
\[
\vx^{(t+1)} = \arg\min_{\AA \vx = \vb} \left \langle \frac{1}{\vert \vx^{(t)} \vert}, \vx^{2} \right \rangle\,,
\]
We abuse notation by applying scalar operations to vectors, with the meaning that they are applied element-wise.

Subsequent works attempted to rigorously analyze this iteration and prove convergence bounds. Oftentimes this relied on specific structure, such as $\vx$ being sparse~\cite{daubechies2010iteratively}. A recent series of works drew inspiration from convergence proofs for the slime-mold dynamics -- a method which essentially solves $\ell_1$ minimization, based on a model used to describe the evolution of a slime mold (\textit{Physarum polycephalum}) as it spreads through its environment in order to optimize its access to food sources~\cite{nakagaki2000intelligence,tero2007mathematical}. Based on the intuition that these dynamics yield a method for solving the transportation problem, Straszak and Vishnoi proved in a series of works~\cite{StraszakV16,StraszakV16Soda,StraszakV16arxiv} that this is as a matter of fact equivalent to the IRLS method, and provided a rigorous convergence analysis for a damped version of it:
\[
\vx^{(t+1)} = \arg\min_{\AA \vx = \vb} \left \langle \frac{1}{\sqrt{ (\vx^{(t)})^2 + \eta^2 }}, \vx^{2} \right \rangle\,.
\]
Unfortunately their convergence proof shows that this method is highly inefficient, and the time to convergence has a high polynomial dependence in the desired accuracy, and the structure of the linear constraint.

By comparison, what we describe in this work is an IRLS method where the weights
are updated according to a thresholding rule. Given a guess $M$ for the optimal value, we perform an iteration equivalent to:
\begin{align*}
c^{(t+1)}_i &= c^{(t)}_i \cdot \psi_{1/(1-\epsilon)}\left(\frac{ x_i^{(t)}/c_i^{(t)} }{\left\langle \frac{1}{\vc^{(t)}}, \left(\vx^{(t)}\right)^2 \right\rangle} \cdot M\right)^2\,, \\
\vx^{(t+1)} &= \arg\min_{\AA \vx = \vb} \left \langle  \frac{1}{\vc^{(t+1)}}, \vx^2 \right\rangle\,,
\end{align*}
where $\psi$ is a thresholding operator i.e. $\psi_b(u) = u$, if $u\geq b$, and $\psi_b(u)=1$ otherwise.
Intuitively, this increases the weights $c_i$ only for the elements where the corresponding component $x_i^2 / c_i$ of the quadratic objective contributes significantly, therefore we want to favor increasing it even more in the future by decreasing the weight $1/c_i$ we place on this coordinate.\footnote{Another way to think of this is that, ignoring the thresholding operator, the update would simply be $c_i^{(t+1)} = (x_i^{(t)})^2 / c_i^{(t)} \cdot \gamma$, where $\gamma$ is some normalization factor. What thresholding achieves here is to decide whether the contribution of a particular coordinate to the energy of the system is sufficiently large compared to the contributions of the entire vector $\vx$.}


\section{Preliminaries}

\subsection{Basic Notation}
\paragraph{Sets.} We let $\R$ be the set of real numbers. For any natural number $n$, we write $[n] := \{1, \dots, n\}$. We denote by $\Delta_m$ the $m$-dimensional simplex i.e. $\Delta_m = \{\vp \in \R^m : \sum_{i=1}^m \vp_i = 1, \vp_i \geq 0 \textnormal{ for all $i$}  \}$.

\paragraph{Vectors.} 
We let $\vzero, \vones \in {\R}^n$ denote the all zeros and all ones vectors, respectively. When it is clear from the context, we apply scalar operations to vectors with the interpretation that they are applied coordinate-wise. 

\paragraph{Matrices.}
We write matrices in bold. We use $\mI$ to denote the identity matrix. 
Given a vector $\vx$ we let $\mdiag(\vx)$ be the diagonal matrix whose entries are given by $\vx$. 
For a symmetric matrix $\AA$, we let $\AA^+$ be its Moore-Penrose pseudoinverse, i.e. $\AA \AA^+ = \AA^+ \AA = \mI_{\textnormal{Im}(\AA)}$. The pseudoinverse can be thought of as replacing all the nonzero eigenvalues of $\AA$ with their reciprocals.

\paragraph{Inner products.} When it is convenient, we use $\langle \cdot, \cdot \rangle$ notation to denote inner products. Given two vectors $\vx, \vy$ of equal dimensions, we let $\langle \vx, \vy \rangle = \vx^\top \vy$. 

\paragraph{Norms.} Given a vector $\vx$, we denote the $\ell_p$ norm of $\vx$ by $\|\vx\|_p = (\sum x_i^p)^{1/p}$. When the subscript is dropped, we refer to the $\ell_2$ norm. From this definition, we can also see that $\|\vx\|_\infty = \max_i \vert x_i\vert$.

\subsection{Proof Technique}
Let us first understand the idea behind our $\ell_\infty$ minimization algorithm.
The problem we aim to solve is $\min_{\vx : \AA\vx = \vb}\|\vx\|_\infty$. Letting $\Delta_m$ be the $m$-dimensional unit simplex, we can write our objective equivalently as 
\begin{align*}
\min_{\vx:\AA\vx = \vb} \|\vx^2\|_\infty 
= \min_{\vx: \AA \vx = \vb} \max_{\vr \in \Delta_m} \langle \vr, \vx^2 \rangle
= \max_{\vr \in \Delta_m} \left(\min_{\vx: \AA\vx = \vb} \langle \vr, \vx^2 \rangle\right) 
:= \max_{\vr \in \Delta_m} \mathcal{E}_{\vr}(\vb)\,,
\end{align*}
where the second identity follows from Sion's theorem~\cite{sion1958general}, which allows us to interchange min and max. The quantity between the parentheses has a very natural interpretation, in the case of electrical networks: it is precisely the electrical energy required to route a demand $\vb$ through an electrical network encoded in $\AA$. Furthermore, we have an easy way to lower bound how this energy increases whenever resistances are increased, which is a finer quantitative version of Rayleigh's monotonicity principle. More precisely, we can easily certify a lower bound on the increase in energy determined by increasing a single coordinate of $\vr$.  Using this observation, which we make more precise in Section~\ref{sec:electric_prelim}, we can identify a set of coordinates of $\vr$ to increase, guaranteeing that if $\vrp$ is the new vector with perturbed resistances, we have
\begin{equation}\label{eq:update_invariant_linfinity}
\frac{\mathcal{E}_{\vrp}(\vb) - \mathcal{E}_{\vr}(\vb)}{\|\vr' - \vr\|_1} \geq M^2\,,
\end{equation}
for a fixed parameter $M$. In the case when no coordinates of $\vr$ can be increased, while preserving this property, this yields a certificate that $\vr$ is as a matter of fact (close to) optimal, and thus we are done (Lemma~\ref{lem:linf_inc_proof}). Hence our goal becomes that of guaranteeing that $\|\vr\|_1$ increases very fast. Indeed, since the "electrical energy" increases at the right rate relative to $\|\vr\|_1$, after the latter has increased sufficiently, we can safely guarantee that $\mathcal{E}_{\vr}(\vb)/\|\vr\|_1 \geq (1-\epsilon)M$, since the increase in $\|\vr\|_1$ cancels out most of the initial error introduced by starting with a potentially poor solution.

The $\ell_1$ minimization algorithm relies on squaring the objective, and then writing it equivalently as 
\begin{align*}
\min_{\vx: \AA\vx = \vb} \|\vx\|_1^2 
=
 \min_{\vx: \AA\vx = \vb} \left(\min_{\vc \in \Delta_m} \left\langle \frac{1}{\vc}, \vx^2 \right\rangle  \right) 
 =
\min_{\vc \in \Delta_m}\left(  \min_{\vx: \AA\vx = \vb}  \left\langle \frac{1}{\vc},\vx^2 \right\rangle  \right)
= \min_{\vc\in\Delta_m} \mathcal{E}_{1/\vc}(\vb)\,.
\end{align*}
For the first identity we used the fact that $\|\vx\|_1^2 = \min_{\vc\in{\Delta}_m}\langle 1/\vc, \vx^2 \rangle$, achieved at $\vc = \vx / \|\vx\|_1$; see~\cite{owen2007robust,sun2012scaled} for further use of this trick.\footnote{Interestingly, this can also be thought of as achieving tightness for reverse H\"older's inequality whenever we are considering the dual `norms' $\ell_{-1}$ and $\ell_{1/2}$.} The second identity follows from joint convexity w.r.t. $\vc$ and $\vx$, which can be verified by computing the Hessian of the function in $(\vx, \vc)$. So completely oppositely from the previous case, the objective of our problem becomes minimizing electrical energy with respect to a set of inverse resistances, which we will call conductances. Note that in this case the quantity that is invariant under scaling $\vc$ by a constant is $\mathcal{E}_{1/\vc} \cdot \|\vc\|_1$. Therefore, equivalently, our goal will be to find the set of conductances $\vc \geq 0$ for which $\left(\mathcal{E}_{1/\vc}\right)^{-1} / \|\vc\|_1 \geq \frac{1}{(1+\epsilon)M}$. 
Similarly to the $\ell_\infty$ case, in this case we make progress by iteratively increasing conductances from $\vc$ to $\vc'$ in such a way that
\begin{equation}\label{eq:update_invariant_lone}
\frac{ \frac{1}{\mathcal{E}_{\vcp}(\vb)}-\frac{1}{\mathcal{E}_{\vc}(\vb)} }{\|\vcp-\vc\|_1} \geq \frac{1}{M^2}\,.
\end{equation}
Just as before, we can prove that unless the value of the objective can not be made smaller than $M$, then $\vc$ can be increased while enforcing this invariant (Lemma~\ref{lem:l1-correct-iterations}). Hence we can prove fast convergence by arguing that $\|\vc\|_1$ increases very fast.

\subsection{Approximate Solutions and Infeasibility Certificates}\label{sec:certificates}
\paragraph{$\ell_\infty$ minimization} We consider the formulation
\begin{equation}
\min_{\AA \vx = \vb} \|\vx\|_\infty\,, \label{eq:linfreg}
\end{equation}
for which we seek an approximate solution in the following sense. Given a target value $M$, we aim to find one of the following:
\begin{enumerate}
\item an approximate solution $\vx$ in the sense that $\AA \vx = \vb$ and $\|\vx\|_\infty \leq (1+\epsilon) M$,
\item an approximate infeasibility certificate $\vr$ in the sense that $\vr \in \Delta_m$ and $\vb^{\top}( \AA \mdiag(\vr)^{-1}  \AA^\top )^+ \vb \geq (1-\epsilon) ^2 M^2$.
\end{enumerate}
We prove in Lemma~\ref{lem:linfcertif} that the latter is indeed an infeasibility certificate.
\begin{lem}\label{lem:linfcertif} 
Let $\vx^*$ be the solution to the problem defined in Equation~\ref{eq:linfreg}, and let $\vr \in \Delta_m$. Then $\|\vx^*\|_\infty^2 \geq  \vb^\top (\AA\mdiag(\vr)^{-1}\AA^{\top})^+ \vb$.
\end{lem}
\begin{proof}
Using Lemma~\ref{lem:energy_charact} we can write
\begin{align*}
\vb^{\top} (\AA\DD(\vr)^{-1} \AA^{\top})^+ \vb 
= \min_{\vx : \AA \vx = \vb} \langle  \vr, \vx^2 \rangle
\leq \langle \vr, (\vx^*)^2 \rangle 
\leq \|\vr\|_1 \|\vx^*\|_\infty^2 
= \|\vx^*\|_\infty^2 \,,
\end{align*}
which gives us what we needed.
\end{proof}

\paragraph{$\ell_1$ minimization} We consider the formulation

\begin{equation}
\min_{\AA \vx = \vb} \|\vx\|_1 \,,\label{eq:l1reg}
\end{equation}
for which seek an approximate solution in the following sense. Given a target value $M$, we seek one of the following:
\begin{enumerate}
\item an approximate infeasibility certificate $\vphi \in \R^n$ in the sense that $\frac{\vb^\top \vphi}{\|\AA^\top \vphi\|_\infty} \geq (1-\epsilon)M$,
\item an approximate 
 feasibility certificate $\vc$ in the sense that $\vc \in \Delta_m$ and ${\vb^{\top}( \AA \mdiag(\vc)  \AA^\top )^+ \vb} \leq (1+\epsilon)^2 M^2$, which yields an approximately feasible solution
  $\vx = \mdiag(\vc) \AA^\top (\AA\mdiag(\vc)\AA^\top)^+ \vb$ in the sense that $\AA \vx = \vb$ and $\|\vx\|_1 \leq (1+\epsilon)M$.
\end{enumerate}
The fact that the former is an approximate infeasibility certificate follows from convex duality. Indeed, one can see that the dual of the minimization problem is $\max_{\vphi : \|\AA^\top \vphi\|_\infty \leq 1} \vb^\top \vphi$, so exhibiting a solution as above implies that the value of this objective is at least $(1-\epsilon) M$.
A proof for the fact that the latter is indeed an approximate feasibility certificate, and that it yields an approximately feasible solution can be found in Lemma~\ref{lem:l1certif}.

\begin{lem}\label{lem:l1certif}
Given $\vc\in \Delta_m$, the vector $\vx = \mdiag(\vc) \AA^\top (\AA\mdiag(\vc)\AA^\top)^+ \vb$  satisfies $\AA \vx = \vb$, and $\|\vx\|_1^2 \leq \vb^\top (\AA \mdiag(\vc) \AA^\top)^+ \vb$.
\end{lem}
\begin{proof}
The fact that $\AA\vx = \vb$ follows directly by substitution, and using the fact that $\vb \in \textnormal{Im}(\AA)$.
Using Lemma~\ref{lem:energy_charact} and the definition in (\ref{eq:energy_def}) we write
\begin{align*}
\vb^\top (\AA \mdiag(\vc) \AA^\top)^+ \vb 
&=
\vb^\top (\AA \mdiag(\vc) \AA^\top)^+ (\AA \mdiag(\vc) \AA^\top) (\AA \mdiag(\vc) \AA^\top)^+ \vb 
\\
&= \sum_{i=1}^m \frac{1}{c_i} \cdot \left(  \mdiag(\vc) \AA^\top (\AA\mdiag(\vc)\AA^\top)^+ \vb \right)^2 \\
&= \sum_{i=1}^m \frac{1}{c_i} \cdot x_i^2
\,.
\end{align*}
We can use this identity inside the following upper bound, which we obtain by applying Cauchy-Schwarz:
\begin{align*}
\|\vx\|_1 = \sum_{i=1}^m { \frac{\vert x_i \vert}{ \sqrt{c_i}} } \cdot \sqrt{c_i} \leq \sqrt{ \left( \sum_{i=1}^m \frac{x_i^2}{c_i} \right) \left( \sum_{i=1}^m c_i \right)   }
=\sqrt{\vb^\top (\AA \mdiag(\vc) \AA^\top)^+ \vb }
\,.
\end{align*}
This yields our claim.
\end{proof}


\section{The Algorithms}\label{sec:descriptions}
Having introduced the necessary notation, we can describe our simple IRLS routine. We prove convergence in Section~\ref{sec:analysis}.

\subsection{The $\ell_\infty$ Minimization Algorithm}\label{subsec:linf_min}
We first present the algorithm for the $\ell_\infty$ version of the problem, since it is the most intuitive. The method attempts to find a weighting of the columns of $\AA$ i.e. a vector $\vr \in \R^m$ for which the corresponding least squares solution has a small $\ell_\infty$ norm; more precisely $\|\vx\|_\infty / \|\vr\|_1 \leq (1+\epsilon)M$ for some chosen target value $M$.

Then the weighting is updated via the following simple thresholding rule. Elements for which the corresponding coordinate of the least squares solution $x_i$ is below the desired target value are left unchanged. The others are scaled exactly by the amount by which the square of the corresponding coordinate $x_i$ violates the desired threshold i.e. $x_i^2 / M^2$.

Note that the iteration defined here simply attempts to construct an infeasibility certificate for the problem defined in Equation~\ref{eq:linfreg}. Building the feasible solution involves maintains a solution obtained by uniformly averaging a subset of the iterates $\vx$ witnessed so far. These are used to return the approximately feasible solution in case the algorithm fails to quickly produce an (approximate) infeasibility certificate. The details referring to how and why we perform this specific set of updates are explained in the convergence proof. The steps involved in building this feasible solution are written in \textcolor{blue}{blue}. They can be ignored if the goal is simply that of returning a yes/no answer.

\begin{algorithm}[tb]
   \caption{$\textsc{$\ell_\infty$-Minimization}(\AA, \vb, \epsilon, M)$}\label{fig:linf}
   \SetAlgoLined
\begin{algorithmic}[1]
   \STATE {\bfseries Input:} Matrix $\AA \in \R^{n\times m}$, vector $\vb \in \R^{n}$, accuracy $\epsilon$, target value $M$.
	\STATE  {\bfseries Output:} Vector $\vx$ such that $\AA\vx = \vb$ and $\|\vx\|_\infty \leq (1+\epsilon) M$, or approximate infeasibility certificate $\vr \in \Delta_m$.
	\STATE $t = 0$, $\vr^{(0)} = \vones/m$.
	\STATE \textcolor{blue}{$t' = 0$, $\vs^{(t')} = \vec{0}$}.
	\WHILE{$\| \vr^{(t)} \|_1 \leq 1/\epsilon$} 
	\STATE $\vx^{(t)} = \arg\min_{\vx : \AA \vx = \vb} \langle \vr, \vx^2 \rangle$.\hfill
\COMMENT{{Equivalently, $\vx^{(t)} = \mdiag(\vr)^{-1} \AA^\top \left(\AA \mdiag(\vr)^{-1} \AA^\top \right)^+ \vb$.}}
\textcolor{blue}{
	\IF{$\| \vx^{(t)} \|_\infty \leq {m^{1/3}} \cdot M$} \label{line:linf_average_or_not}
	\STATE $t' = t'+1$, $\vs^{(t')} = \vs^{(t'-1)} + \vx^{(t)}$. 
	\ENDIF
	\IF{$\| \vs^{(t')}  \|_\infty / t' \leq (1+\epsilon) M$}
	\STATE {\bfseries return $\vs^{(t')} / t'$}.
	\ENDIF
	}
	\STATE  $\alpha^{(t)}_i =
 \begin{cases} 
 1 & \textnormal{if } \vert x^{(t)}_i \vert  < (1+\epsilon) M, 
\, \\ 
 \frac{ (x^{(t)}_i)^2  }{M^2} & \textnormal{otherwise. }
 \end{cases}$
 	\IF{$\valpha^{(t)} = \vones$}
	\STATE {\bfseries return $\vx^{(t)}$. }
	\ENDIF
	\STATE $\vr^{(t+1)} = \vr^{(t)} \cdot \valpha^{(t)}$.
	\STATE $t = t+1$.
	\ENDWHILE
	\STATE {\bfseries return} $\vr^{(t)} /\|\vr^{(t)}\|_1$. 
\end{algorithmic}
\end{algorithm}

\paragraph{Correctness.}
We notice that Algorithm~\ref{fig:linf} has two possible outcomes. Either it returns a primal approximately feasible vector (lines~\REFlineretone and~\REFlinerettwo), or returns a dual certificate (line~\REFlineretthree).
In the former case, it is clear from the description of the algorithm that the returned vector is indeed approximately feasible: line~\REFlineretone\, returns a uniform average of vectors satisfying the linear constraint with small $\ell_\infty$ norm; line~\REFlineretthree\, returns the $\vx^{(t)}$ computed within the corresponding iteration, whenever $\valpha^{(t)} =  \vec{1}$, i.e. $\|\vx^{(t)}\|_\infty < (1+\epsilon) M$. 

Also, note that in case none of these stopping conditions is triggered, the algorithm returns a dual certificate on line~\REFlineretthree\, after a finite number of iterations. Indeed, note that every iteration where ${\alpha}_i^{(t)} \neq \vec{1}$, at least one element of $\vr^{(t)}$ gets increased by a factor of at least $(1+\epsilon)^2$, due to way $\valpha^{(t)}$ is defined. Since the algorithm stops when $\|\vr^{(t)}\|_1 = 1/\epsilon$, no element of $\vr$ can be scaled more than $O(\log_{(1+\epsilon)} (m/\epsilon))$ times, hence the total number of iterations is very roughly upper bounded by $O(m\log(m/\epsilon)/\epsilon)$. We will see in Section~\ref{sec:analysis} that we can prove a much finer upper bound. 

Finally, we need to argue that whenever the algorithm returns on line~\REFlineretthree, it returns an infeasibility certificate as per Lemma~\ref{lem:linfcertif}. We defer the proof to Lemma~\ref{lem:linf_inc_proof} in  Section~\ref{sec:analysis}.

\subsection{The $\ell_1$ Minimization Algorithm}\label{subsec:lone_min}

The $\ell_1$ version is very similar. As a matter of fact, it can be re-derived simply by attempting to solve the convex dual of the problem from (\ref{eq:linfreg}), which is an $\ell_\infty$ minimization problem, by using the routine from Figure~\ref{fig:linf}. However, since the reduction requires several, and previous works attempted to solve this directly using various versions of IRLS, we provide a natural iteration which does not involve any reductions.

\begin{algorithm}[tb]
   \caption{$\textsc{$\ell_1$-Minimization}(\AA, \vb, \epsilon, M)$}\label{fig:lone}
   \SetAlgoLined
\begin{algorithmic}[1]
   \STATE {\bfseries Input:} Matrix $\AA \in \R^{n\times m}$, vector $\vb \in \R^{n}$, accuracy $\epsilon$, target value $M$.
	\STATE  {\bfseries Output:} Vector $\vx$ such that $\AA\vx = \vb$ and $\|\vx\|_1 \leq (1+\epsilon) M$, or approximate infeasibility certificate $\vphi \in \Delta_n$.
	\STATE $t = 0$, $\vc^{(0)} = \vones/m$.
	\STATE \textcolor{blue}{$t' = 0$, $\vs^{(t')} = \vzero$, ${\vPhi}^{(0)} = \vzero$}.
	\WHILE{$\| \vc^{(t)} \|_1 \leq 1+\frac{1}{(1+\epsilon)^2-1}$}
	\STATE
	$\vphi^{(t)} = \left(\AA \mdiag(\vc) \AA^\top  \right)^+ \vb$.\hfill
\COMMENT {
 Equivalently, $\vphi^{(t)}$ is the vector of potentials which induce the electrical flow $\vx = \arg\min_{\AA\vx=\vb} \langle 1/\vc, \vx^2 \rangle$ via $\vx = \mdiag(\vc) \AA^\top \vphi$.
}
\textcolor{blue}{
	\IF{$\left\| \frac{\AA^\top \phi^{(t)}}{\vb^\top \vphi^{(t)}} \right\|_\infty \leq {m^{1/3}} \cdot \frac{1}{M}$} 
	\STATE $t' = t'+1$, $\vs^{(t')} = \vs^{(t'-1)} + \left\vert \frac{\AA^\top \vphi^{(t)}}{\vb^\top \vphi^{(t)}}\right\vert$, $\vPhi^{(t')} = \vPhi^{(t'-1)} + \frac{\vphi^{(t)}}{\vb^\top \vphi^{(t)}}$. 
	\ENDIF
	\IF{$\|  \vs^{(t')}  \|_\infty / t' \leq \frac{1}{(1-\epsilon) M}$}
	\STATE {\bfseries return $\vPhi^{(t')} / t'$}.
	\ENDIF
	}
	\STATE  $\alpha^{(t)}_i =
 \begin{cases} 
 1 & \textnormal{if }
  \frac{
\vert \AA^\top \vphi^{(t)} \vert_i
 }{
   \vb^\top \vphi^{(t)}
 } 
 \leq \frac{1}{(1-\epsilon) M}, 
\, 
\\ 
\left(
 \frac{ 
(\AA^\top \vphi^{(t)})_i
 }{
 \vb^\top \vphi^{(t)}
 }
 \right)^2 \cdot M^2 & \textnormal{otherwise. }
  \\ 
 \end{cases}$
 	\IF{$\valpha^{(t)} = \vones$}
	\STATE {\bfseries return $\vphi^{(t)}$. }
	\ENDIF
	\STATE $\vc^{(t+1)} = \vc^{(t)} \cdot \valpha^{(t)}$.
	\STATE $t = t+1$.
	\ENDWHILE
	\STATE {\bfseries return} $\vx = \mdiag(\vc^{(t)}) \AA^\top \vphi^{(t)}$.%
\end{algorithmic}
\end{algorithm}

\paragraph{Correctness.}

We notice that Algorithm~\ref{fig:lone} has two possible outcomes. Either it returns an approximate infeasibility certificate (lines~\REFlineloneretone\, and~\REFlinelonerettwo), or returns an approximately feasible solution (line~\REFlineloneretthree).

Let us verify that in the former case the returned vector is indeed an approximate infeasibility certificate. Line~\REFlineloneretone\, returns ${\vPhi}^{(t')} = \sum_{t \in S} \frac{\vphi^{(t)}}{ \vb^\top \vphi^{(t)} }$, where we know that $S$ is a set for which 
\begin{align*}
&\left\Vert \AA^\top {\vPhi}^{(t')} \right\Vert_\infty
=
\left\Vert \AA^\top \cdot \sum_{t\in S}  \frac{  \vphi^{(t)}}{\vb^\top \vphi^{(t)}} \right\Vert_\infty
\\
&=
\left\Vert \sum_{t\in S}  \frac{ \AA^\top \vphi^{(t)}}{\vb^\top \vphi^{(t)}} \right\Vert_\infty
\leq
\left\Vert \sum_{t\in S} \left\vert \frac{\AA^\top \vphi^{(t)}}{\vb^\top \vphi^{(t)}} \right\vert \right\Vert_\infty \leq \frac{t'}{(1-\epsilon) M}\,.
\end{align*}
Since $\vb^\top {\vPhi}^{(t')} = t'$, we see that returned vector ${\vPhi}^{(t')} / t'$ is an approximate infeasibility certificate, as defined in Section~\ref{sec:certificates}. If the algorithm returns on line~\REFlinelonerettwo, we get that $\left\Vert  \frac{\AA^\top \vphi^{(t)}}{\vb^\top \vphi^{(t)}} \right\Vert_\infty \leq \frac{1}{(1-\epsilon)M}$, hence $\vphi^{(t)}$ is an approximate infeasibility certificate.

Also, note that in case none of these stopping conditions is triggered, the algorithm returns a solution on line~\REFlineloneretthree\, after a finite number of iterations. Indeed, just as in the $\ell_\infty$ case, in every iteration some conductance gets increased by a factor of at least $\Omega(1+\epsilon)$, hence the algorithm must stop in finite time. We provide a rigorous analysis of the time required for convergence in Section~\ref{sec:analysis}.

Finally, we need to argue that whenever the algorithm returns a solution on line~\REFlineloneretthree, it is indeed an approximately feasible solution. We defer the proof to Lemma~\ref{lem:l1-correct-iterations} in  Section~\ref{sec:analysis}.


\section{The Algorithm Analyses}\label{sec:analysis}

\subsection{The Flow/Potential Interpretation}\label{sec:flowpotinter}
While we study a very general problem, it is very useful to develop intuition based on the case where $\AA$ is the vertex-edge incidence matrix of a graph. In this case we will always think of the sought solution $\vx$ as a flow on the graph's edges. The corresponding dual object is a set of potentials $\vphi$ defined on the graph's vertices.

To be more precise, we consider the following setting. Let $G = (V,E)$ be an undirected graph. For each edge we choose an arbitrary orientation, and define $E^+(v)$ be the set of arcs leaving vertex $v$, and $E^-(v)$ the set of arcs entering vertex $v$, for all $v$.

Letting $m = \vert E \vert$, $n = \vert V \vert$, we consider the matrix $\AA \in \R^{n \times m}$ where
\[
\AA_{ve} = 
\begin{cases} 
+1 & \textnormal{if } e \in E^+(v), \\  
-1 & \textnormal{if } e \in E^-(v), \\  
0 &\textnormal{otherwise.}
\end{cases}
\]

One can easily verify that given a vector $\vx \in \R^m$ defined on the arcs of the graph (which we will think of as a \textbf{flow}), after applying the operator $\AA$ we obtain the demand routed by this flow $\AA \vx \in \R^n$, which lives in the space of \textbf{potentials} defined on the graph's vertices.

Therefore the $\ell_\infty$ minimization problem from \ref{eq:linfreg} can be interpreted as finding the flow $\vx$ with minimum congestion which routes the demand $\vb$, while the $\ell_1$ minimization problem from \ref{eq:l1reg} corresponds to finding the minimum cost flow routing the demand $\vb$.

With this interpretation in mind, we proceed to define some objects that in the case of electrical networks correspond to energy and electrical flows.

We use weightings of $\AA$'s columns $\vc \in \R^m$ which we refer to as \textbf{conductances}. We equivalently refer to the reciprocals $\vr \in \R^m$, with $\vr_i = 1/\vc_i$, which we call \textbf{resistances}. Our analysis is exclusively based on tracking a potential function which corresponds to the electrical energy of a flow.

\begin{defn}[Energy of a flow]
Given a flow $\vx \in \R^m$, along with a vector of resistances $\vr \in \R^m$, we let the energy of $\vx$ be
\[
\mathcal{E}_{\vr}(\vx) = \langle \vr, \vx^2 \rangle\,.
\]
Overloading this notation, given a vector $\vb \in \R^n$, let the \textbf{electrical energy} be
\begin{equation}\label{eq:energy_def}
\mathcal{E}_{\vr}(\vb) = \min_{\vx : \AA \vx = \vb} \mathcal{E}_{\vr}(\vx)\,,
\end{equation}
in other words this is the minimum energy over all flows satisfying $\AA \vx = \vb$. We drop the argument whenever $\vb$ is clear from the context.
\end{defn}

\subsection{Preliminaries on Electrical Energy}\label{sec:electric_prelim}
Throughout the paper, our analyses will rely on a potential function, which in the case of resistor networks corresponds to the electrical energy. In this section we provide a few useful facts.

\begin{lem}[Characterization of Electrical Energy]\label{lem:energy_charact}
Given a vector of resistances $\vr \in \R^m$, we have the following equivalent characterizations for the electrical energy.
\begin{align}
\mathcal{E}_{\vr}(\vb) &= \vb^\top \left(  \AA \mdiag(\vr)^{-1} \AA^\top \right)^+ \vb \\
&=\max_{\phi} 2\cdot \vb^{\top} \vphi  - \sum_{i=1}^m \frac{\left(\AA^{\top} \vphi \right)_i ^2}{  r_i } \label{eq:energy_max_char} \\
&= \left( \min_{\vphi : \vb^{\top} \vphi = 1} \sum_{i=1}^m \frac{\left(\AA^{\top} \vphi\right)_i^2}{ r_i }\right)^{-1}\,.\label{eq:energy_inverse_char}
\end{align}
Furthermore, if $\vx$ is the minimizing flow for the expression in (\ref{eq:energy_def}), and $\vphi$ is the maximizing set of potentials for the expression in (\ref{eq:energy_max_char}), then for all $i$:
\begin{equation}
x_i = (\AA^\top \vphi)_i / r_i\,.  \label{eq:flow_pot_equival}
\end{equation}
\end{lem}
Since the proof is standard, we defer it to Section~\ref{sec:pf_l1}.

As a corollary, we can derive a lower bound on the increase in energy after increasing resistances.

\begin{lem}\label{lem:lb_inc_en} Let $\vr, \vrp$, and let $\vx = \arg\min_{\vx : \AA \vx = \vb} \langle \vr, \vx^2\rangle$. Then, one has that
\begin{align*}
\mathcal{E}_{\vrp}(\vb) \geq \mathcal{E}_{\vr}(\vb) + \sum_{i=1}^m r_i x_i^2 \left(1 - \frac{r_i}{r'_i}\right)\,.
\end{align*}
\end{lem}
\begin{proof}
We use the characterization from Equation~\ref{eq:energy_max_char} for characterizing electrical energy. Let $\vphi$ be the argument that maximizes $(\ref{eq:energy_max_char})$ for resistances $\vr$. We certify a lower bound on 
$\mathcal{E}_{\vrp}(\vb)$ using $\vphi$ as follows:
\begin{align*}
\mathcal{E}_{\vrp}(\vb) 
&\geq 
2 \cdot \vb^\top \vphi - \sum_{i=1}^m \frac{\left(\AA^\top \vphi\right)_i^2}{r'_i}
\\
&= 2 \cdot \vb^\top \vphi - \sum_{i=1}^m \frac{\left(\AA^\top \vphi\right)_i^2}{r_i}
+ \sum_{i=1}^m \frac{\left(  \AA^\top \vphi \right)_i^2 }{r_i} \cdot \left(1 - \frac{r_i}{r'_i}\right)
\\
&= \mathcal{E}_{\vr}(\vb) + \sum_{i=1}^m \frac{\left(\AA^\top \vphi\right)_i^2}{r_i} \cdot \left(1- \frac{r_i}{r'_i}\right)
\,.
\end{align*}
Finally substituting the relation between flows and potentials from Lemma~\ref{lem:energy_charact}, Equation (\ref{eq:flow_pot_equival}), we obtain the desired claim.
\end{proof}

We can derive a similar lower bound on the inverse energy, after increasing conductances.
\begin{lem}\label{lem:energy_inverse_increase}
Let $\vphi = \arg\min_{\phi: \langle \vb, \vphi \rangle = 1} \langle \vc, (\AA^\top \vphi)^2\rangle$. Then one has that
\begin{align*}
\frac{1}{\mathcal{E}_{1/\vcp}(\vb)} \geq \frac{1}{\mathcal{E}_{1/\vc}(\vb)}
+ \frac{1}{\mathcal{E}_{1/\vc}(\vb)^2} 
\cdot \sum_{i=1}^m {c_i} (\AA^\top \vphi)_i^2 \left(1-\frac{c_i}{c'_i}\right)\,.
\end{align*}
\end{lem}
\begin{proof}
We use the following basic inequality: for $x,x'> 0$ one has $\frac{1}{x'} \geq \frac{1}{x} + \frac{x-x'}{x^2}$, which follows from $(x - x')^2 \geq 0$.
Also, from the definition of energy in (\ref{eq:energy_def}), we obtain an upper bound on the new energy, after perturbing conductances. Let $\vx = \arg\min_{\AA \vx = \vb} \langle 1/\vc, \vx^2\rangle$, i.e. the electrical flow corresponding to conductances $\vc$. We therefore have:
\begin{align*}
\mathcal{E}_{1/\vcp}(\vb) \leq \sum_{i=1}^m \frac{1}{c'_i} x_i^2
= \sum_{i=1}^m \frac{1}{c_i} x_i^2 + \sum_{i=1}^m \frac{1}{c_i} x_i^2 \cdot \left(\frac{c_i}{c'_i} - 1\right) 
= \mathcal{E}_{1/\vc}(\vb) + \sum_{i=1}^m \frac{1}{c_i}x_i^2 \left(\frac{c_i}{c'_i}-1\right)\,.
\end{align*}
Using the fact that by optimality, $x_i = c_i (\AA^\top \vphi)_i$ (per Lemma~\ref{lem:energy_charact}), and combining with the previous inequality we obtain
\begin{align*}
\frac{1}{\mathcal{E}_{1/\vcp}(\vb)} 
&\geq\frac{1}{\mathcal{E}_{1/\vc}(\vb)} +
 \frac{1}{\mathcal{E}_{1/\vc}(\vb)^2} \cdot
 \left(\mathcal{E}_{1/\vc}(\vb) - \mathcal{E}_{1/\vcp}(\vb)^2\right)
 \\
&\geq \frac{1}{\mathcal{E}_{1/\vc}(\vb)}
+ \frac{1}{\mathcal{E}_{1/\vc}(\vb)^2} 
\cdot \sum_{i=1}^m {c_i} (\AA^\top \vphi)_i^2 \left(1-\frac{c_i}{c'_i}\right)\,,
\end{align*}
which is what we wanted.
\end{proof}

\subsection{Convergence Proof for $\ell_\infty$ Minimization}\label{sec:conv_pf_linf}
Having put together all these tools, we are ready to analyze the algorithms presented in Section~\ref{sec:descriptions}.
We first prove that \textsc{$\ell_\infty$-Minimization} returns a correct infeasibility certificate, whenever it returns on line~\REFlineretthree. This lemma is key to understanding the intuition behind the algorithm.
\begin{lemma}\label{lem:linf_inc_proof}
Whenever \textsc{$\ell_\infty$-Minimization} returns on line~\REFlineretthree, $\vr/\|\vr\|_1$ is a correct approximate infeasibility certificate in the sense that
$$\mathcal{E}_{\vr/\|\vr\|_1}(\vd) \geq (1-\epsilon)^2 M^2\,.$$
\end{lemma}
\begin{proof}
First notice that by Lemma~\ref{lem:linfcertif}, the lower bound on energy is indeed an approximate infeasibility certificate. Now we proceed to prove that throughout the iterations of the algorithm, energy increases at the right rate.

We show that every iteration satisfies the invariant
\begin{equation}
\frac{
\mathcal{E}_{\vr^{(t+1)}}(\vd)
-
\mathcal{E}_{\vr^{(t)}}(\vd)}
{\| \vr^{(t+1)} - \vr^{(t)} \|_1} \geq M^2\,.\label{eq:inc_ratio}
\end{equation}
This is easy to verify using Lemma~\ref{lem:lb_inc_en}, which lower bounds the increase in energy after perturbing resistances. We see that using the perturbation rule defined on line~\REFlinelinfpertrule\, of the algorithm, energy increases as follows
\[
\mathcal{E}_{\vr^{(t+1)}}(\vd) \geq \mathcal{E}_{\vr^{(t)}}+\sum_{i=1}^m  r^{(t)}_i(x_i^{(t)})^2 \cdot \left(1- \frac{1}{\alpha_i^{(t)}}\right)\,.
\]
For every coordinate of $\vr^{(t)}$ that has changed we see that the ratio between the contribution to above lower bound of that specific coordinate, and the increase in resistance is
\[
\frac{  r_i^{(t)}(x_i^{(t)})^2\left(1 - \frac{1}{\alpha_i^{(t)}}\right)}{r_i \left( \alpha_i^{(t)} -1 \right)} = \frac{(x_i^{(t)})^2}{\alpha_i^{(t)}} = M^2\,.
\]
Therefore, summing up over all coordinates we obtain the desired inequality.
Finally, we notice that initially $\mathcal{E}_{\vr^{(0)}}(\vd) \geq 0$, and $\|\vr^{(0)}\|_1 = 1$. So once $\|\vr^{(t)}\|_1 \geq \frac{1}{\epsilon}$, one has that, using (\ref{eq:inc_ratio}),
\begin{align*}
\frac{
\mathcal{E}_{\vr^{(t)}}(\vd)
-
\mathcal{E}_{\vr^{(0)}}(\vd)}
{\| \vr^{(t)} \|_1 - 1} \geq M^2\,,
\end{align*}
and thus
\[
\mathcal{E}_{\vr^{(t)}}(\vd) \geq M^2 (\|\vr^{(t)}\|_1 - 1),
\]
and equivalently:
\[
\mathcal{E}_{\vr^{(t)} / \|\vr^{(t)}\|_1} (\vd) \geq M^2 \left(1  - \frac{1}{ \|\vr^{(t)}\|_1 }\right) \geq M^2 (1-\epsilon)  \,,
\]
which implies what we needed.
\end{proof}

Knowing that the algorithm is correct, we can now proceed and prove that it converges fast (convergence rate can be slightly improved by using a more careful schedule for $M$ and $\epsilon$; we defer this improvement to Section~\ref{sec:phase_improve}). 
\begin{lemma}\label{lem:linf_conv}
The algorithm $\textsc{$\ell_\infty$-Minimization}$ returns a solution after $O(m^{1/3} \log(1/\epsilon) / \epsilon + \log(m/\epsilon)/\epsilon^2)$ iterations.
\end{lemma}
\begin{proof}
We show that unless the algorithm returns an approximately feasible solution
on lines~\REFlineretone~ or~\REFlinerettwo, then there exists a coordinate $i \in [m]$ for which  $r_i$ increases very fast.

Suppose the algorithm has run for $T$ iterations without returning an approximately feasible solution. Consider the partial sum of iterates obtained so far $\vs^{(t')}$ for some $t' \leq T$. Since the algorithm did not return on line~\REFlineretone, we know that $\|\vs^{(t')}\|_\infty / t' \geq (1+\epsilon) M$. Therefore there exists a coordinate $i \in [m]$ for which $\vs^{(t')}_i \geq (1+\epsilon) M t'$. In other words, letting $I$ be the set of iterates that have contributed to $\vs^{(t')}$, one definitely has that
\[
\sum_{t \in I} \vert \vx^{(t)} \vert \geq t' \cdot (1+\epsilon) M\,,
\]
and thus
\[
\sum_{t \in I} \sqrt{ \alpha_i^{(t)} } \geq t' \cdot (1+\epsilon)\,,
\]
where we used the fact that for each iteration $t\in I$ one has that $\sqrt{\alpha_i^{(t)}} = \vert x_i^{(t)} \vert M$ due to the perturbation rule defined on line~\REFlinelinfpertrule.
This implies that restricting ourselves only to iterations where $\alpha_i$ increased the corresponding resistance $r_i$, we have that
\begin{equation}\label{eq:linf_p1}
\sum_{t \in I, \alpha_i^{(t)} > 1} \sqrt{ \alpha_i^{(t)} } \geq t' \epsilon\,,
\end{equation}
By the condition on line~\REFlinelinfaverageornot\, we see that for all iterations $t\in I$, one has 
\begin{equation}\label{eq:linf_p2}
\sqrt{\alpha_i^{(t)}} \leq {m^{1/3}}\,.
\end{equation}
Also since we only consider the iterations $t\in I$ with $\alpha_i^{(t)} > 1$, the rule from line~\REFlinelinfpertrule\, also enforces that for all these iterations
\begin{equation}\label{eq:linf_p3}
\sqrt{\alpha_i^{(t)}} \geq 1+\epsilon\,.
\end{equation}
Equations (\ref{eq:linf_p1}), (\ref{eq:linf_p2}) and (\ref{eq:linf_p3}) suggest that the product $\prod_{t\in I, \alpha_i^{(t)} > 1} \sqrt{\alpha_i^{(t)}}$ increases very fast: intuitively the worst case should occur either when all the factors contribute equally, either all of them are as small as possible (i.e. $1+\epsilon$, or as large as possible, i.e. $m^{1/3}$). We formalize this intuition in Lemma~\ref{lem:prod_lb}, which implies that
\[
\prod_{t\in I, \alpha_i^{(t)} > 1} \sqrt{\alpha_i^{(t)}} \geq \min \left\{ \left({m^{1/3}}\right)^{\frac{t'\epsilon }{ m^{1/3}}}, (1+\epsilon)^\frac{t' \epsilon}{1+\epsilon} \right\}\,.
\]
Hence setting 
\[
t' \geq 10 \left(\frac{m^{1/3}\log(1/\epsilon) }{\epsilon} + \frac{\log(m/\epsilon)}{\epsilon^2}\right)
\]
suffices to lower bound this product by $\sqrt{m/\epsilon}$.
Since each iteration a resistance ${r}_i^{(t)}$ gets multiplied by the corresponding $\alpha_i^{(t)}$, and all resistances are initially $1/m$, this lower bound implies that $r^{(t)}_i \geq 1/\epsilon$. But this means that the algorithm will finish execution after the current iteration, according to the condition on line~\REFlinewhile.

Finally, we need to upper bound the number of iterations not in $I$; these correspond to those iterations where $\|\vx^{(t)}\|_\infty \geq m^{1/3} \cdot M$, so there exists some index $i$ for which $\alpha^{(t)}_i \geq m^{2/3}$. Therefore some resistance gets multiplied by $m^{2/3}$. Since all resistances are initially $1/m$, in the worst case, each such iteration increases one resistance from $1/m$ to $m^{-1/3}$. Therefore this can happen at most $m^{1/3}\log(1/\epsilon)/\epsilon$ times, before the sum of resistances becomes at least $1/\epsilon$, and the algorithm finishes.

Combining these two cases, we obtain our bound.
\end{proof}

We can prove the convergence bound for $\textsc{$\ell_1$-Minimization}$ similarly. The main difference is that this time we maintain conductances, and the potential function that enables us to prove convergence is $1/\mathcal{E}_{1/\vc}$.

\subsection{Convergence Proof for $\ell_1$-Minimization}\label{sec:defpflone}
\begin{lemma}\label{lem:l1-correct-iterations}
Whenever \textsc{$\ell_1$-Minimization} returns on line~\REFlineloneretthree, $\vc/\|\vc\|_1$ is a correct approximate feasibility certificate in the sense that
\[
\frac{1}{\mathcal{E}_{\|\vc\|_1/\vc}} \geq \frac{1/(1+\epsilon)^2}{M^2}\,.
\]
\end{lemma}
\begin{proof}
By Lemma~\ref{lem:l1certif}, this also yields a solution $\vx$ such that $\AA \vx = \vb$ and $\|\vx\|_1 \leq \sqrt{\mathcal{E}_{\|\vx\|_1 / \vc}} \leq M(1+\epsilon)$.

In order to prove that at the end of the execution the $\ell_1$ norm of this solution is small enough, this time we track as potential function the inverse energy $1/\mathcal{E}_{1/\vc}$. More precisely, we show that every iteration satisfies the invariant
\begin{equation}\label{eq:l1_inc_ratio}
\frac{ \frac{1}{\mathcal{E}_{\vc^{(t+1)}}(\vd)} - \frac{1}{\mathcal{E}_{\vc^{(t)}}(\vd)} }{ \|\vc^{(t+1)} - \vc^{(t)}\|_1 } \geq \frac{1}{M^2}\,.
\end{equation}
This is easy to verify using Lemma~\ref{lem:energy_inverse_increase}, which lower bounds the increase in inverse energy after perturbing conductances. We see that using the perturbation rule defined on line~\REFlinelonepertrule\, of the algorithm, inverse energy increases as follows
\begin{align*}
\frac{1}{\mathcal{E}_{1/\vc^{(t+1)}}(\vd)} 
\geq
\frac{1}{ \mathcal{E}_{1/\vc^{(t)}}}
+
\frac{1}{ \mathcal{E}_{1/\vc^{(t)}}^2} \cdot
\sum_{i=1}^m  {c_i^{(t)}} (\AA^\top \vphi^{(t)})_i^2 \cdot \left(1- \frac{1}{\alpha_i^{(t)}}\right)\,.
\end{align*}
For every coordinate of $\vc^{(t)}$ that has changed we see that the ratio between the contribution to above lower bound of that specific coordinate, and the increase in conductance is
\begin{align*}
\frac{1}{ \mathcal{E}_{1/\vc^{(t)}}^2 } \cdot \frac{  {c_i^{(t)}}(\AA^\top \vphi^{(t)})_i^2\left(1 - \frac{1}{\alpha_i^{(t)}}\right)}{c_i \left( \alpha_i^{(t)} -1 \right)} 
= 
\frac{(\AA^\top \vphi^{(t)})_i^2}{ \mathcal{E}_{1/\vc^{(t)}}^2 } \cdot 
\frac{1}{\alpha_i^{(t)}} 
= 
\left(\frac{(\AA^\top \vphi^{(t)})_i}{\vb^\top \vphi^{(t)} }\right)^2 \cdot 
\frac{1}{\alpha_i^{(t)}} 
=
\frac{1}{M^2}\,,
\end{align*}
where we used the fact that $\vb^\top \vphi^{(t)} = \mathcal{E}_{1/\vc^{(t)}}$ (Lemma~\ref{lem:energy_charact}).

Therefore, summing up over all coordinates we obtain the desired inequality. Since $\mathcal{E}_{1/\vc^{(0)}}\geq 0$ and $\|\vc^{(0)}\|_1 = 1$, we know that once $\|\vc^{(t)}\|_1 \geq 1+\frac{1}{(1+\epsilon)^2 - 1} = O( \frac{1}{\epsilon} )$, one has that, using (\ref{eq:l1_inc_ratio}),
\[
\frac{  \frac{1}{\mathcal{E}_{1/\vc^{(t)}}(\vd)}-\frac{1}{\mathcal{E}_{1/\vc^{(0)}}(\vd)} }{\|\vc^{(t)}\|_1-1} \geq \frac{1}{M^2}
\,
\]
and thus,
\[
\frac{1}{\mathcal{E}_{1/\vc^{(t)}}(\vd)} 
\geq \frac{\|\vc^{(t)}\|_1-1}{M^2},
\]
and equivalently:
\[
 \,\textnormal{}\\
\mathcal{E}_{\|\vc^{(t)}\|_1/\vc^{(t)}}(\vd) 
= \mathcal{E}(1/\vc^{(t)})(\vd)  \cdot \|\vc^{(t)}\|_1 
\leq M^2 \cdot \frac{ \|\vc^{(t)}\|_1}{\|\vc^{(t)}\|_1 - 1} \leq M^2 (1+\epsilon)^2\,,
\]
which is what we needed.
\end{proof}

Next we prove that the algorithm converges fast. Convergence rate can be slightly improved by using a more careful schedule for $M$ and $\epsilon$, which we defer to Section~\ref{sec:phase_improve}.
\begin{lemma}\label{lem:lone_conv_fast} The algorithm \textsc{$\ell_1$-Minimization} returns a solution after $O(m^{1/3} \log (1/\epsilon)/\epsilon + \log(m/\epsilon)/\epsilon^2)$ iterations.
\end{lemma}
\begin{proof}
The proof follows the lines of the proof we used for Lemma~\ref{lem:linf_conv}:
unless the algorithm returns an approximate infeasibility certificate on lines~\REFlineloneretone\, or~\REFlinelonerettwo, then there exists a coordinate $i\in[m]$ for which $c_i$ increases very fast.

Suppose the algorithm has run for $T$ iterations without returning an approximate infeasibility certificate.
Consider the partial sum of iterates obtained so far $\vs^{(t')}$ for some $t'\leq T$. Since the algorithm did not return on line~\REFlineloneretone, we know that $\|\vs^{(t')}\|_\infty /t' \geq \frac{1}{(1-\epsilon)M}$, therefore there exists a coordinate $i\in[m]$ for which $s^{(t')}_i \geq t' \cdot \frac{1}{(1-\epsilon) M}$. In other words, letting $I$ be the set of iterates that contributed to $\vs^{(t')}$, one has that
\[
s^{(t')}_i
= \sum_{t\in I} \left\vert 
\frac{
(\AA^\top \vphi^{(t)})_i
}
{
\langle \vb, \vphi^{(t)} \rangle
} 
\right\vert 
\geq t'\cdot \frac{1}{(1-\epsilon) M}
\]
and  thus, since by definition $\alpha_i^{(t)} = \left( \frac{ (\AA^\top \vphi^{(t)})_i }{\vb^\top \vphi^{(t)}} \right)^2 \cdot M^2$,
\[
\sum_{t\in I}\sqrt{\alpha_i^{(t)}} \geq t' \cdot \frac{1}{1-\epsilon}\,.
\]
Therefore, restricting ourselves only to iterations where $\alpha_i$ increased the corresponding conductance $c_i$, we have that
\begin{equation}\label{eq:blah1}
\sum_{t\in I, \alpha_i^{(t)} > 1} \sqrt{\alpha_i^{(t)}} \geq t' \cdot \left(\frac{1}{1-\epsilon} - 1\right) \geq t' \cdot \epsilon\,.
\end{equation} 
By the condition on line~\REFlineloneaverageornot\, we see that for all iterations $t \in I$, one has 
\begin{equation}\label{eq:blah2}
\sqrt{\alpha_i^{(t)}} \leq m^{1/3}\,.
\end{equation}
So considering only the iterations $t \in I$ with $\alpha_i^{(t)} > 1$, the rule from line~\REFlinelonepertrule\, also enforces that for all these iterations 
\begin{equation}\label{eq:blah3}
\sqrt{\alpha_i^{(t)}} \geq \frac{1}{1-\epsilon}\,.
\end{equation}
Combining Equations (\ref{eq:blah1}), (\ref{eq:blah2}), and (\ref{eq:blah3}), and applying Lemma~\ref{lem:prod_lb}, exactly the same way we did in the proof of Lemma~\ref{lem:linf_conv} implies that
\[
\prod_{t \in I, \alpha_i^{(t)} > 1} \sqrt{\alpha_i^{(t)}} \geq 
\min\left\{
\left(m^{1/3}\right)^{\frac{t'\epsilon}{m^{1/3}}},
\left( \frac{1}{1-\epsilon} \right)^{\frac{t'\epsilon}{1/(1-\epsilon)}}
\right\}
\]
So if 
\[
t' \geq 10 \left(\frac{m^{1/3}\log(1/\epsilon) }{\epsilon} + \frac{\log(m/\epsilon)}{\epsilon^2}\right)
\]
once again we have that this product is lower bounded by $\sqrt{m \cdot \left(1 + \frac{1}{(1+\epsilon)^2 - 1}\right)}$, therefore the corresponding conductance $c_i^{(T)} \geq 1 + \frac{1}{(1+\epsilon)^2 - 1}$, since its initial value was $1/m$.
Since we can only control the total number of iterations $T$, we can lower bound $t'$ by showing that the number of iterations not in $I$ can not be too large. Just as before, we lower bound the number of iterations where $\left\|\frac{\AA^\top \vphi^{(t)}}{\vb^\top \vphi^{(t)}}  \right\|_\infty \geq m^{1/3} / M$. Note that whenever this happens, there exists an index $i$ for which $\alpha_i^{(t)} \geq m^{2/3}$. Therefore some conductance gets multiplied by $m^{2/3}$. Again, using an identical argument to the one from the proof of Lemma~\ref{lem:linf_conv}, we see that this can not happen more than $O(m^{1/3} \log(1/\epsilon) /\epsilon)$ times. Combining this with the sufficient number of iterations required by the other case, we obtain our bound.
\end{proof}


\section{Using Phases to Improve the Iteration Count}\label{sec:phase_improve}
In this section, we show that via minor modifications to our algorithms, we can improve the number of iterations to 
$
O\left(\frac{m^{1/3} \log(1/\epsilon)}{\epsilon^{2/3}} + \frac{\log m}{\epsilon^2}\right)
$
thus obtaining the bound promised by Theorem~\ref{thm:main_theorem}.
This relies on the observation that the entire difficulty of the problem is concentrated on improving the quality of a solution from $(1-2\epsilon)M$ to $(1-\epsilon)M$. For conciseness, let us focus on the $\ell_\infty$ case, and consider the convergence argument described in Sections~\ref{sec:conv_pf_linf}. Our goal there is to increase the sum of resistances to $1/\epsilon$, since our argument assumes that the initial energy could be arbitrarily small.

 However, if we assume that we warm start the method with a set of resistances $\vr_0$, $\|\vr_0\|_1 = 1$, for which the corresponding energy is already large enough, $\mathcal{E}_{\vr_0} \geq (1-2\epsilon)^2 M^2$, we only need to iterate until we obtain a set of resistances $\vr$ such that $\|\vr\|_1 = 3$ (rather than $1/\epsilon$) in order to certify that the current energy/resistance ratio is as large as desired, i.e. $\mathcal{E}_{\vr}/ \|\vr\|_1 \geq (1-\epsilon)^2 M^2$. This in turn improves the number of iterations the algorithm needs before it returns. We expand these ideas in what follows.
 
Now, suppose we have a set of resistances $\vr_0$, such that $\|\vr_0\|_1 = 1$ and $\mathcal{E}_{\vr_0} \geq (1-2\epsilon)^2 M^2$. Let us analyze the number of iterations of the method described in Section~\ref{sec:conv_pf_linf} that we require before we can return $\vr$ such that $\mathcal{E}_{\vr} / \|\vr\|_1 \geq (1-\epsilon)^2 M^2$ or a solution $\vx$ such that $\|\vx\|_\infty \leq (1+\epsilon) M$.

First, we claim that if each update satisfies the invariant from Equation (\ref{eq:inc_ratio}), then we can stop iteration once $\|\vr\|_1 = 3$. Indeed, in this case, one has that
\begin{align*}
\frac{\mathcal{E}_{\vr}}{\|\vr\|_1} &= \frac{\mathcal{E}_{\vr_0} + \left(\mathcal{E}_{\vr}-\mathcal{E}_{\vr_0}\right)}{\|\vr_0\|_1 + \|\vr-\vr_0\|_1}\\
&\geq \frac{(1-2\epsilon)^2 M^2 + \|\vr-\vr_0\|_1 M^2}{1 + \|\vr-\vr_0\|_1} \\
&\geq M^2 \left(1 - \frac{4\epsilon}{1 + \|\vr-\vr_0\|_1}\right) \\
&\geq M^2 (1-\epsilon)^2\,,
\end{align*}
whenever $\|\vr-\vr_0\|_1 \geq 3$.

The remaining analysis is carried over almost identically, except that the threshold set on line~\REFlinelinfaverageornot\, is changed to $\rho = \epsilon m$, and our goal is to get $\prod_{t\in I, \alpha_i^{(t)}>1}\alpha_i \geq \sqrt{3 m}$.

For the iterations that satisfy this threshold, by applying Lemma~\ref{lem:prod_lb} we see that it is sufficient to witness a small number $t'$ of such iterations such that 
\[
\min \left\{ {\rho}^{\frac{t'\epsilon }{ \rho}}, (1+\epsilon)^\frac{t' \epsilon}{1+\epsilon} \right\} \geq \sqrt{3m}\,,
\]
so $t' = \Theta\left(\frac{\rho}{\epsilon}\cdot \frac{\log m}{\log \rho} + \frac{\log m}{\epsilon^2}\right)$ suffices.

For the iterations that do not satisfy the threshold, in the worst case, each of them increases one resistance from $1/m$ to $\rho^2 / m$ so this can happen at most $O(m/\rho^2)$ times.

Setting $\rho = \left(\epsilon m\right)^{1/3}$ we get that the total number of iterations is at most $O\left( \frac{m^{1/3}}{\epsilon^{2/3}} \log(1/\epsilon) + \frac{\log m}{\epsilon^2} \right)$.

All of this holds assuming we have a good warm start for resistances. We obtain it by recursively invoking the same method for target value $(1-1.75 \epsilon)^2 M^2$, and $.25 \epsilon$ accuracy. In case of failure, this returns a vector $\vx$ which certainly satisfies $\|\vx\|_\infty \leq M$, so this concludes the entire run on the algorithm; otherwise, it returns a certificate consisting of resistances for which the ratio between the corresponding energy and $\ell_1$ norm is at least $(1-2\epsilon)^2 M^2$, so they can be used as a warm start.

Recursion ends once $\epsilon \geq 1/2$. We note that since the desired accuracy gets increased by a constant factor after each level of recursion, the total number of iterations is dominated by those performed at the top level (i.e. for the lowest $\epsilon$).  Hence our result. 

Note that this method can also be implemented slightly more naturally by running Algorithm~\ref{fig:linf} with a varying schedule for $M$ and $\epsilon$. 

Improving the number of iterations for $\ell_1$ minimization is done analogously.

\section{From Approximate Decision to Approximate Optimization}\label{sec:appx_opt}
Our algorithms are designed to solve an approximate decision problem, given a guess for the value of the optimal solution. While this follows from a standard reduction, for the sake of completeness we prove here that this is sufficient to optimize the problem approximately without paying more than an additional constant overhead in running time.

To be more specific, let us first focus on $\ell_\infty$ minimization. Theorem~\ref{thm:main_theorem} states that given a guess $M$ and accuracy $\epsilon$, the algorithm either returns an approximately feasible solution with value $\|\vx\|_\infty \leq (1+\epsilon)M$, or an infeasibility certificate certifying that $\|\vx^*\|_\infty \geq (1-\epsilon)M$. Hence this restricts the search interval for the true value either within the interval $[0, (1+\epsilon)M]$ or $[(1-\epsilon)M, \infty)$.

We initialize our search interval to $[\|\vx_0\|_2/\sqrt{m}, \|\vx_0\|_\infty]$ where $\vx_0$ is the initial iterate obtained with uniform resistances. Using Lemma~\ref{lem:linfcertif} we easily verify that $\|\vx_0\|_2/\sqrt{m}$ is indeed a lower bound on $\|\vx^*\|_\infty$, since energy lower bounds the squared optimal value.

Given a search interval $[L, U]$, we let $M = \sqrt{LU}$, $\tilde{\epsilon} = \min\left\{\frac{1}{2}, \left(\frac{U}{L}\right)^{1/6}-1 \right\}$. We invoke Theorem~\ref{thm:main_theorem} for target value $M$ and accuracy $\tilde{\epsilon}$. Depending on the outcome we update the search interval to $[L, (1+\tilde{\epsilon})M]$ or $[(1-\tilde{\epsilon})M, U]$.

 When $U/L \leq 1+\epsilon/4$ we stop the search, call the algorithm for target value $U(1+\epsilon/4)$ and accuracy $\frac{\epsilon/4}{1+\epsilon/4}$, then output the approximately feasible iterate returned by the algorithm. The fact that this call indeed returns an approximately feasible iterate follows from the fact that $U$ is certainly feasible, since this is an invariant maintained by our search, and that if the algorithm were to return an infeasibility certificate it must have needed that $U(1+\epsilon/4)(1-\frac{\epsilon/4}{1+\epsilon/4}) < U$, which is false. Thus we know that the returned solution has value at most $U(1+\epsilon/4)(1+\frac{\epsilon/4}{1+\epsilon/4}) \leq L(1+\epsilon)$, so it satisfies the desired approximation guarantee.

Finally, we analyze the cost of the search. Note that each iteration of the search reduces $\log U - \log L$ be a constant factor, and it stops whenever it becomes at most $\log(1+\epsilon/4) = \Theta(\epsilon)$.
For as long as $U/L > (3/2)^6$, the algorithm is invoked with accuracy $1/2$, and $\log U - \log L$ gets reduced by a constant factor, so this happens at most $O(\log \log m)$ times. Once $U/L$ becomes small enough, i.e. $\log U - \log L < 6 \log(3/2)$, we use accuracy $\exp((\log U-\log L)/6)-1 = \Theta(\log(U/L))$. Note that from Theorem~\ref{thm:main_theorem} we know that the number of iterations of the algorithm for a single invocation depends on $1/\tilde{\epsilon}^c$, where $c$ is a fixed constant; due to our schedule for choosing $\tilde{\epsilon}$, the total cost of this sequence of invocations is dominated by the cost of the final one, where $\tilde{\epsilon} = \Theta(\epsilon)$.

So letting $\mathcal{T}(\epsilon)$ be the number of iterations required by the algorithm from Theorem~\ref{thm:main_theorem} to solve the approximate decision problem to accuracy $\epsilon$, we have that solving the approximate optimization problem requires $O\left(\mathcal{T}(1/2) \log \log m + \mathcal{T}(\epsilon)\right)$ iterations. 

The $\ell_1$ minimization problem is treated similarly, so we omit its description.


\newcommand{\rulesep}{}

\begin{figure*}[h!]
\subfigure[][$\log(\textnormal{iterations})$ for $\ell_\infty$\\ minimization as function of $\log(\epsilon^{-1})$.]{\label{subfig:a}
\begin{minipage}[t]{0.236\textwidth}
	\centering
	\includegraphics[width=\textwidth]{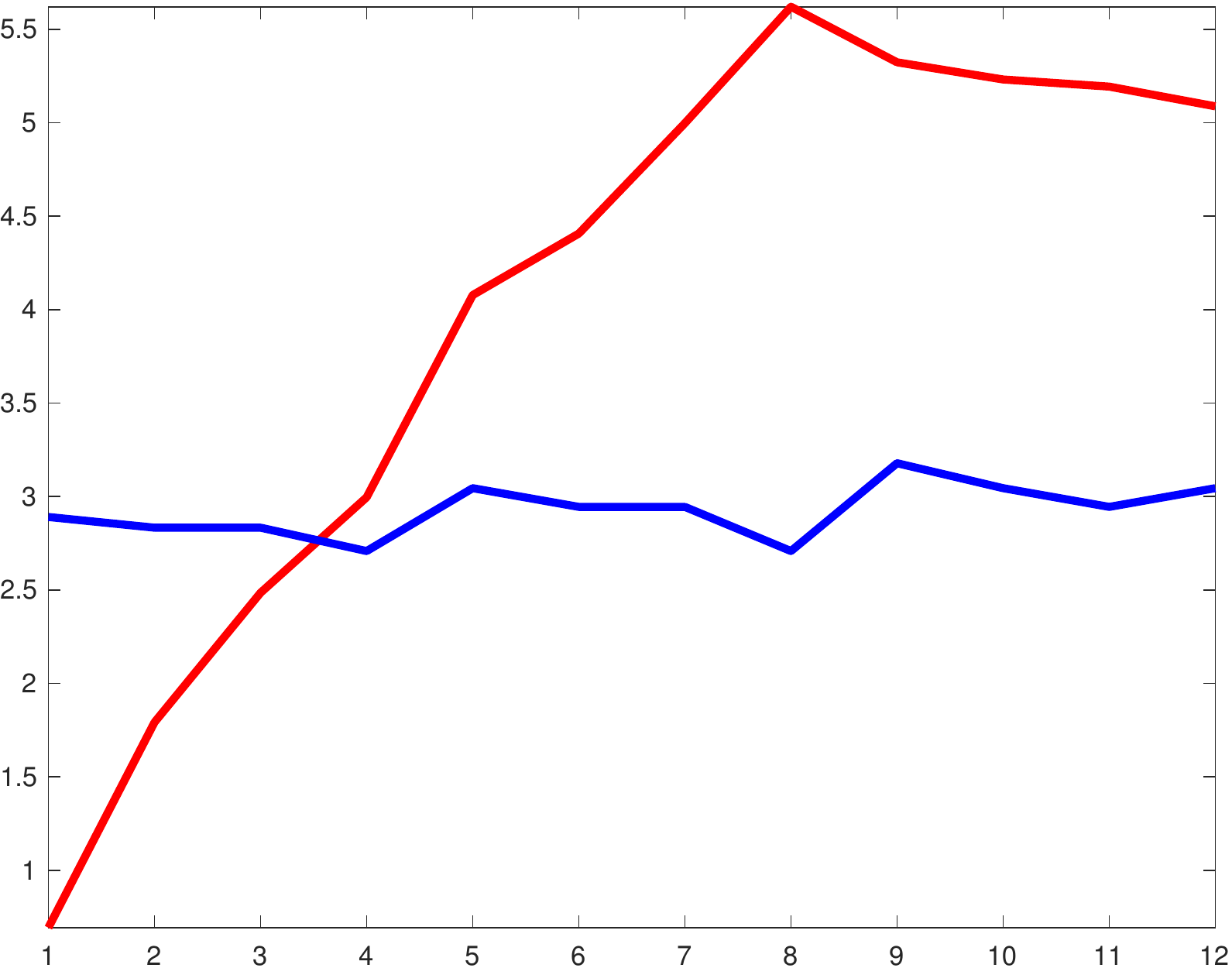}
\end{minipage}
}\rulesep
\subfigure[][Time (sec) for $\ell_\infty$\\ minimization as function of $\log(\epsilon^{-1})$.]{\label{subfig:b}
\begin{minipage}[t]{0.236\textwidth}
	\centering
	\includegraphics[width=\textwidth]{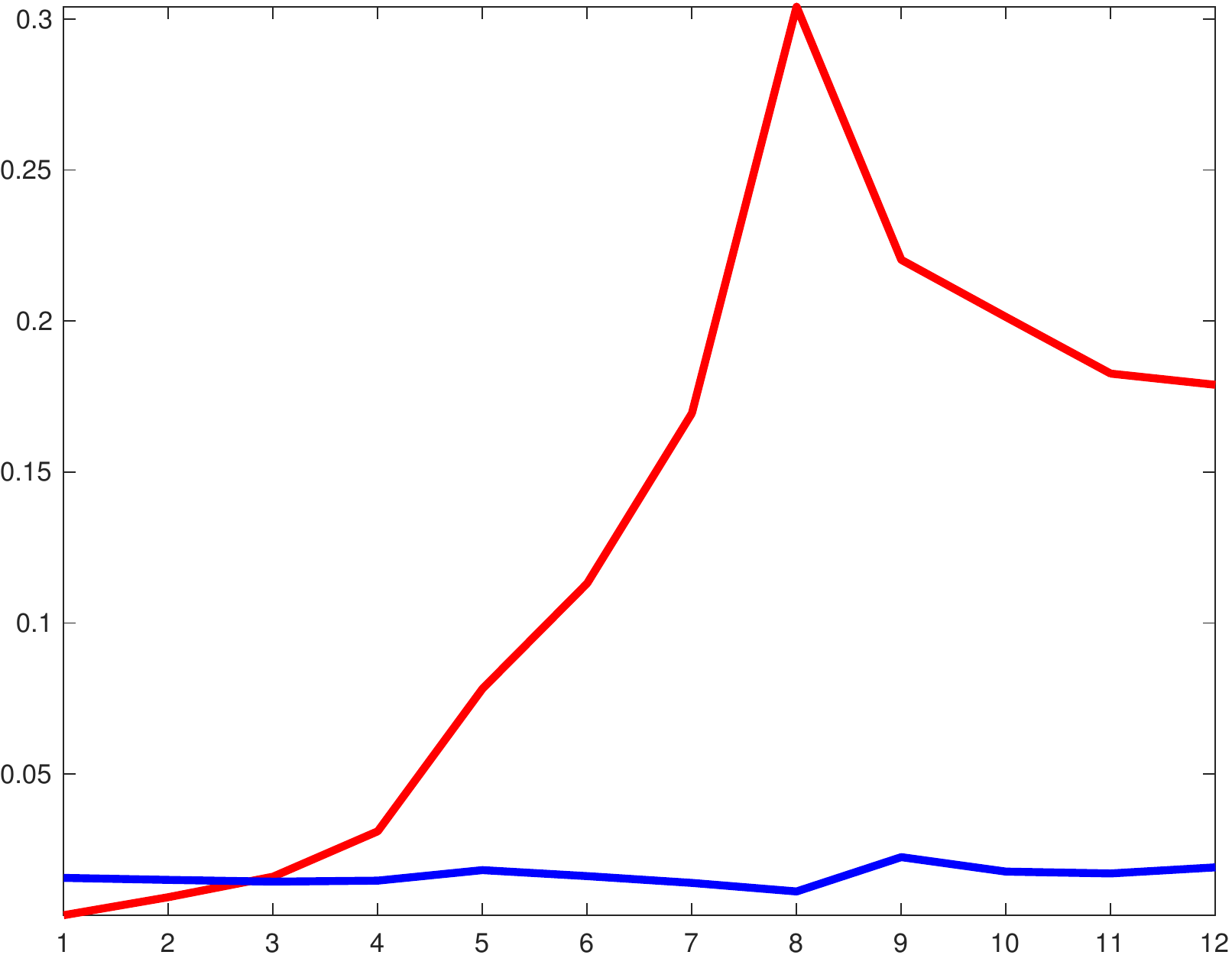}
\end{minipage}
}\rulesep
\subfigure[][Number of iterations for $\ell_\infty$\\ minimization as function of $m/200$.]{\label{subfig:c}
\begin{minipage}[t]{0.236\textwidth}
	\centering
	\includegraphics[width=\textwidth]{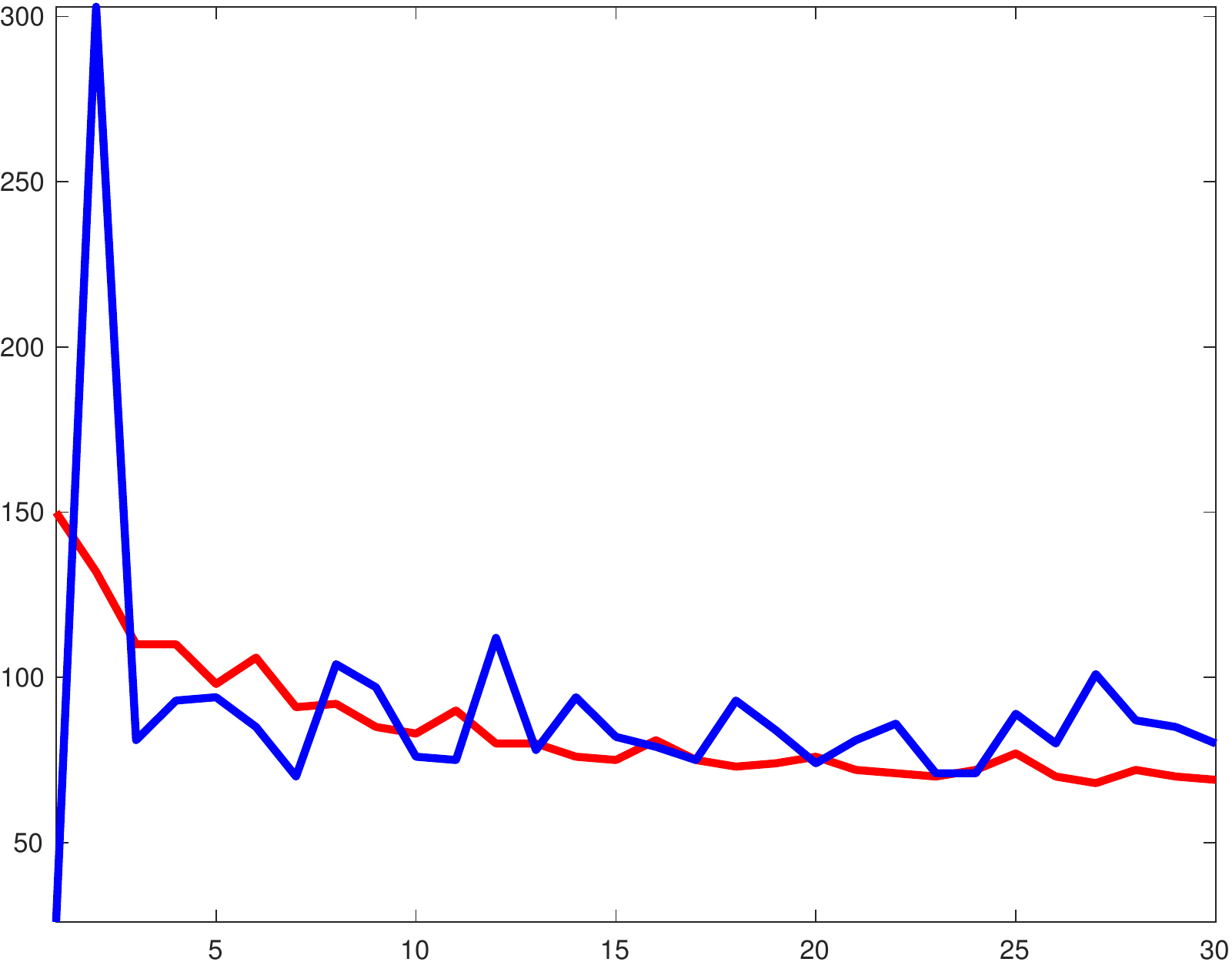}
\end{minipage}
}\rulesep
\subfigure[][Time (sec) for $\ell_\infty$\\ minimization as function of $m/200$.]{\label{subfig:d}
\begin{minipage}[t]{0.236\textwidth}
	\centering
	\includegraphics[width=\textwidth]{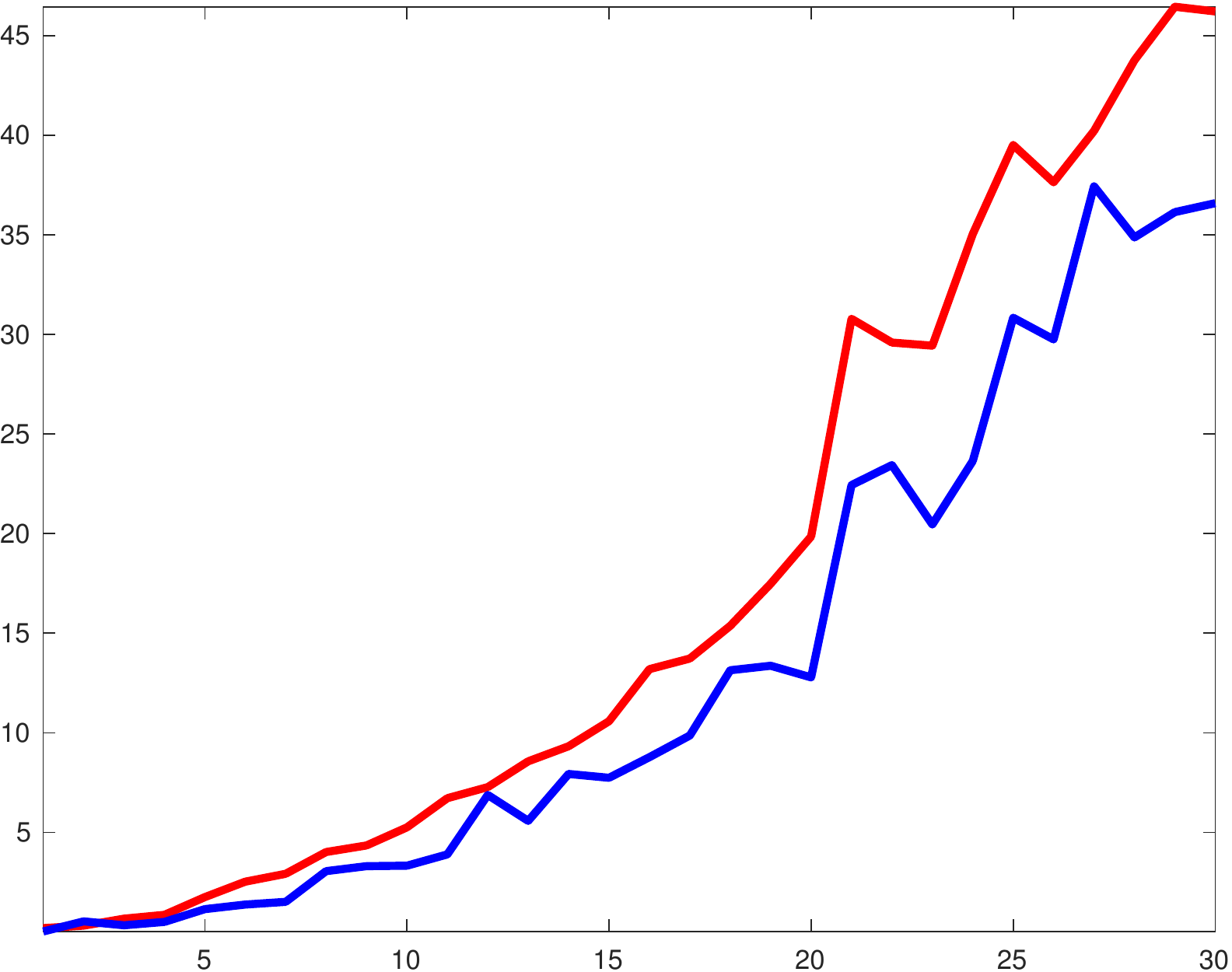}
\end{minipage}
}
\subfigure[][$\log(\textnormal{iterations})$ for $\ell_1$\\ minimization as function of $\log(\epsilon^{-1})$.]{\label{subfig:e}
\begin{minipage}[t]{0.236\textwidth}
	\centering
	\includegraphics[width=\textwidth]{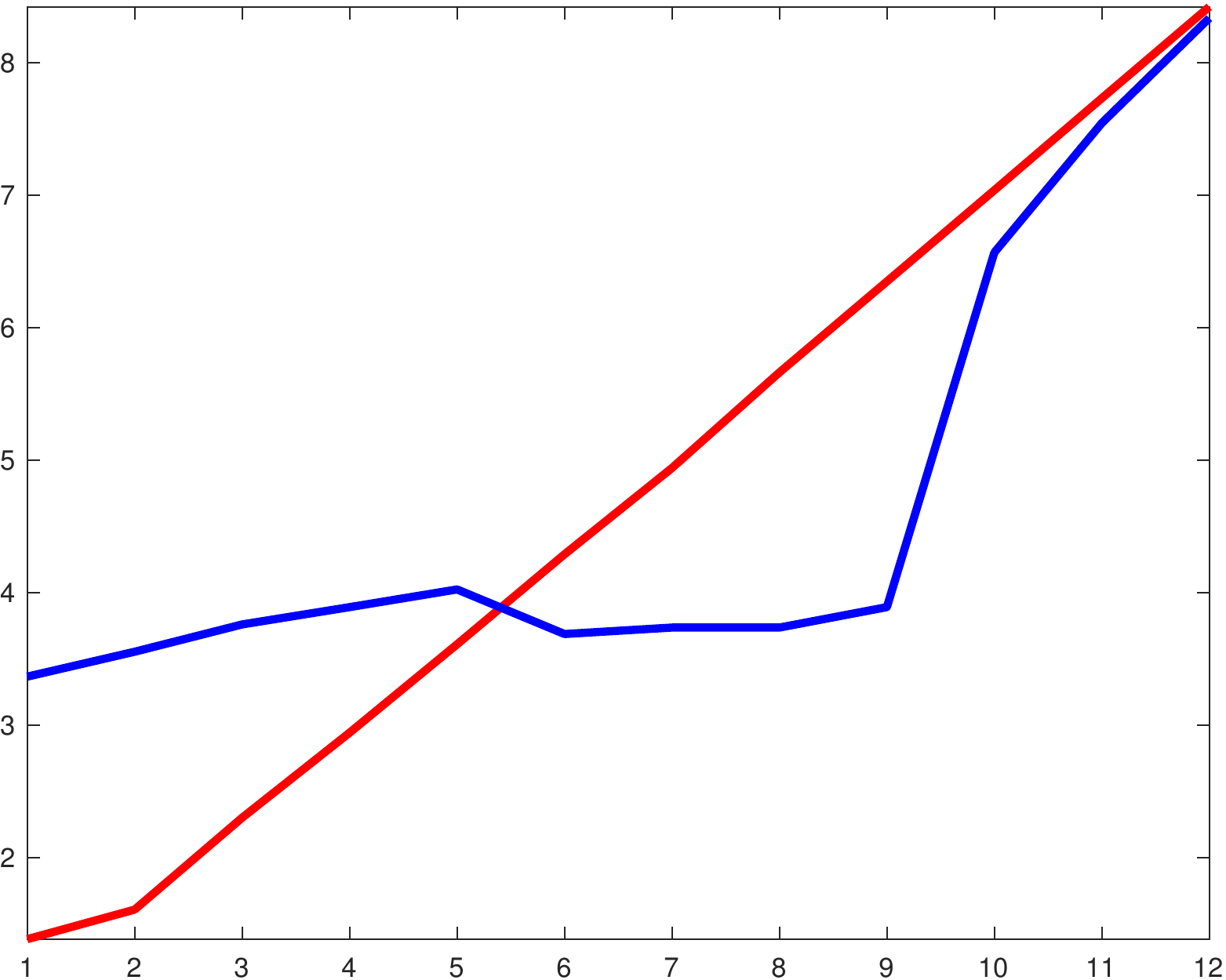}
\end{minipage}
}\rulesep
\subfigure[][Time (sec) for $\ell_1$\\ minimization as function of $\log(\epsilon^{-1})$.]{\label{subfig:f}
\begin{minipage}[t]{0.236\textwidth}
	\centering
	\includegraphics[width=\textwidth]{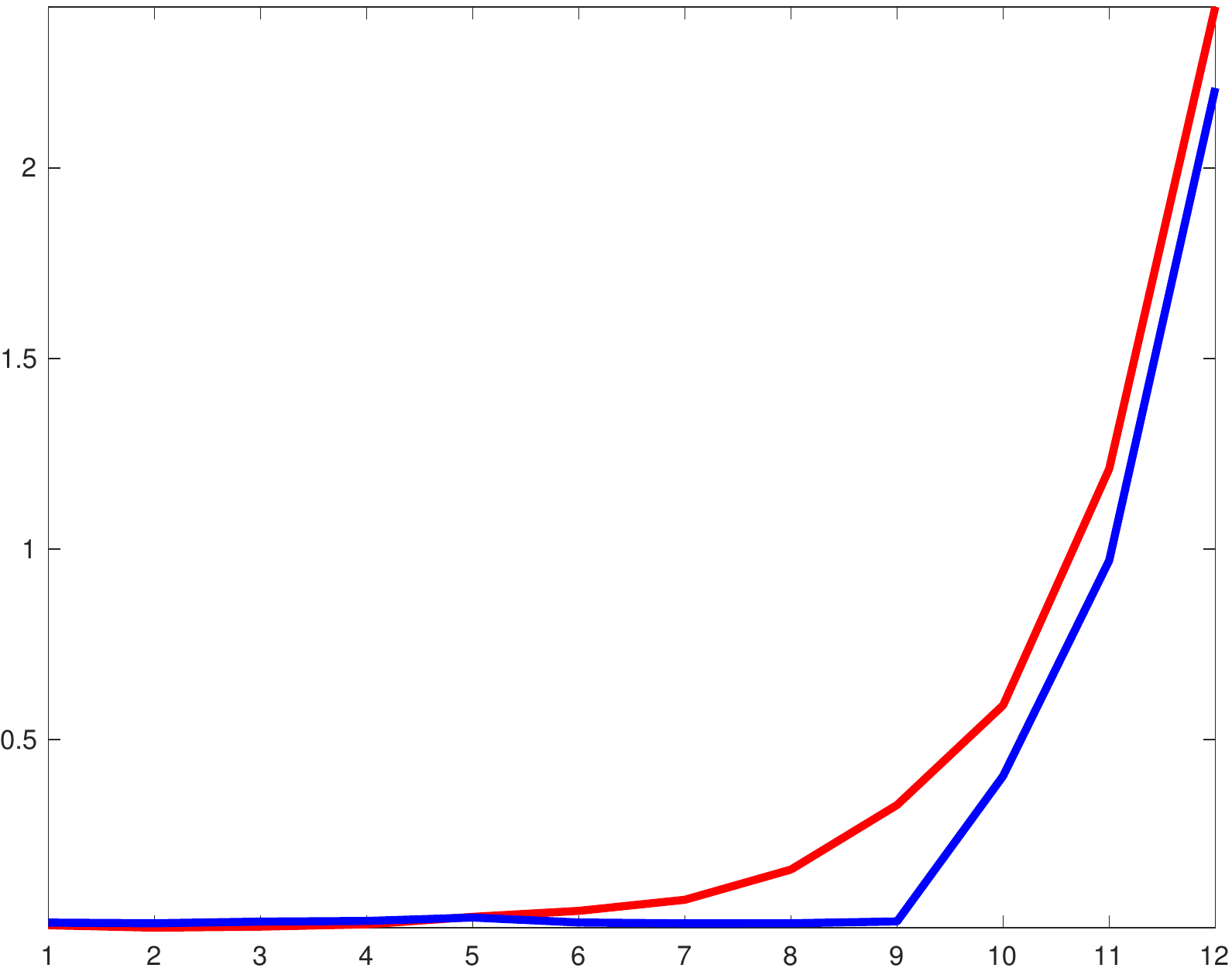}
\end{minipage}
}\rulesep
\subfigure[][Number of iterations for $\ell_1$\\ minimization as function of $m/200$.]{\label{subfig:g}
\begin{minipage}[t]{0.236\textwidth}
	\centering
	\includegraphics[width=\textwidth]{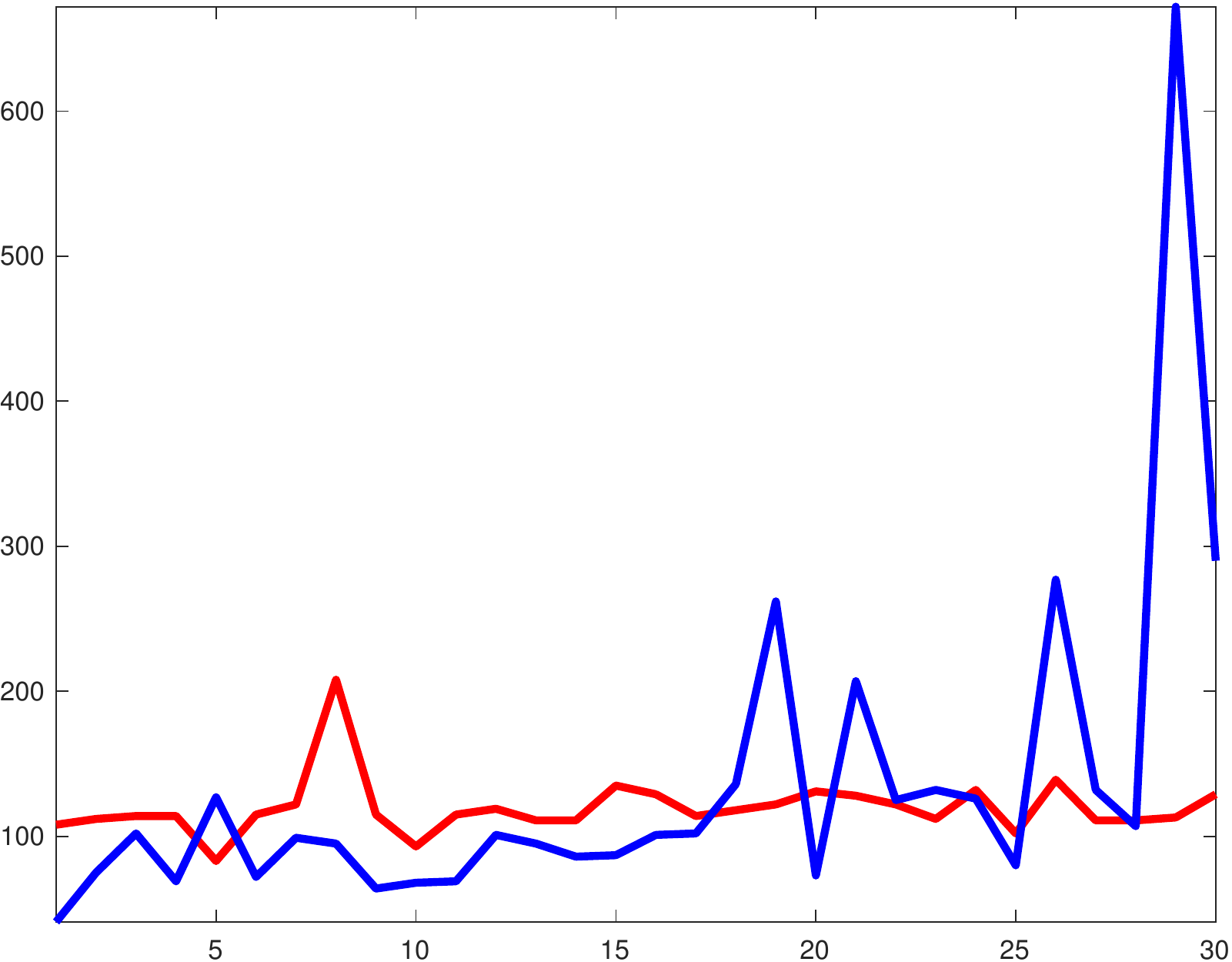}
\end{minipage}
}\rulesep
\subfigure[][Time (sec) for $\ell_1$\\ minimization as function of $m/200$.]{\label{subfig:h}
\begin{minipage}[t]{0.236\textwidth}
	\centering
	\includegraphics[width=\textwidth]{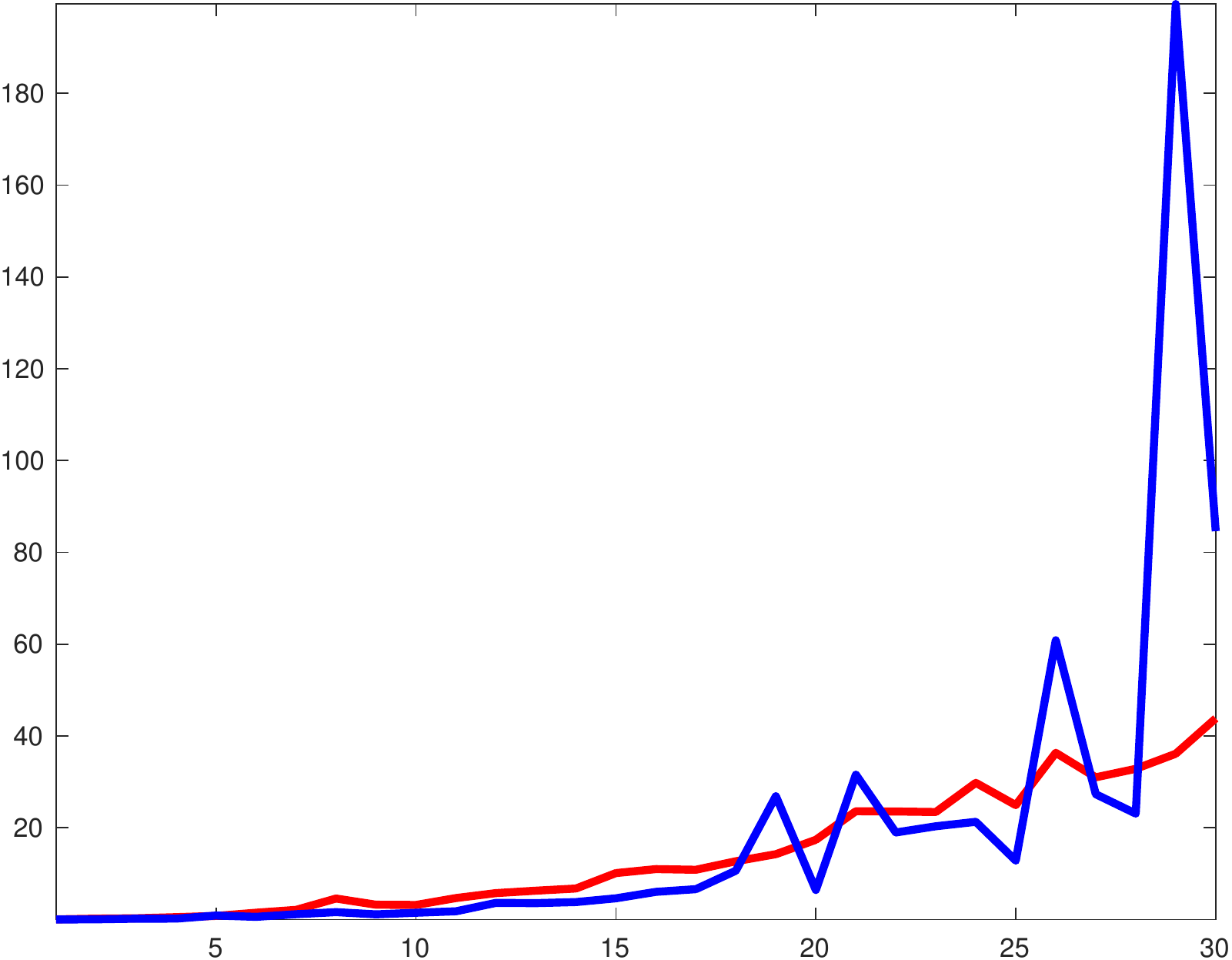}
\end{minipage}
}
\caption{Experimental results.}
\label{fig:experiments}
\end{figure*}

\section{Experiments}\label{sec:experiments}
We test both our resistance/conductance update schemes in order to verify that the resulting algorithms converge fast in practice.
We slightly modify the schemes such that they always update their target value $M$ depending on the value of the objective they have achieved so far. We stop when given the history of witnessed iterates, we can certify a sufficiently small duality gap. For solving linear systems, we used the conjugate gradient implementation from the $\ell_1$-MAGIC optimization suite~\cite{cgsolvem}.

We test both algorithms while varying $\epsilon$, and varying $m$. We consider both the update scheme given by our algorithms from Section~\ref{sec:descriptions}, and one where we attempt to double the length of the step for as long as the invariants from (\ref{eq:update_invariant_linfinity}) and (\ref{eq:update_invariant_lone}), respectively, are maintained. We notice that in general, using this long step strategy, we improve both the number of iterations and the running time.

The plots corresponding to the standard update scheme (short-steps) are drawn in \textcolor{red}{red}, those corresponding to the long-step version are drawn in \textcolor{blue}{blue}.

The experiments where we varied $\epsilon$ are reported in figures~\ref{subfig:a},~\ref{subfig:b},~\ref{subfig:e}, and ~\ref{subfig:f}.
For all these experiments, the input consists a random $150 \times 200$ matrix $\AA$ with orthogonal rows, and a vector $\vb$ obtained from applying $\AA$ to a 
 $\pm 1$-vector of sparsity $15$.
We plot the number of iterations/running time of the algorithm for $\epsilon = 1/2^k$, where $k \in \{1, \dots, 12\}$. 

We notice that for these experiments, the number of iterations for the short-step version does indeed scale linearly with $\epsilon^{-1}$; the long-step version makes significant gains in the $\ell_\infty$ case.

The experiments where we varied $m$ are reported in figures~\ref{subfig:c},~\ref{subfig:d},~\ref{subfig:g}, and~\ref{subfig:h}. For all  these experiments, the input consists of a random $150\times (200\cdot k)$ matrix $\AA$ with orthogonal vectors, and a  vector $\vb$ obtained from applying $\AA$ to a 
 $\pm 1$-vector of sparsity $15$, and a fixed accuracy $\eps = .01$.
We plot the number of iterations required by the algorithm for $k \in \{1, \dots, 30\}$.

We notice that for these experiments, both the number of iterations and the running time scale significantly better than by $m^{1/3}$, which suggests that this polynomial dependence in $m$ depends on the input structure, and can be avoided in practice.

\section*{Acknowledgements} 
 
AE was partially supported by NSF CAREER grant CCF-1750333 and NSF grant CCF-1718342. AV was partially supported by NSF grant CCF-1718342.

\clearpage
\appendix

\section{Deferred Proofs}

\subsection{Proof of Lemma~\ref{lem:energy_charact}}\label{sec:pf_l1}
\begin{proof}
We can write the formulation from (\ref{eq:energy_def}) as an unconstrained optimization problem using Lagrange multipliers:
\begin{align*}
\mathcal{E}_{\vr}(\vb) = \min_{\AA \vx = \vb} \langle \vr, \vx^2 \rangle
&= \min_{\vx}\max_{\vphi}  \langle \vr, \vx^2 \rangle 
+ 2 \langle \vphi, \vb - \AA \vx \rangle 
\\
&= \max_{\vphi} \min_{\vx} \langle \vr, \vx^2 \rangle 
+ 2\langle \vphi, \vb - \AA \vx \rangle\,.
\end{align*}
By making the gradient with respect to $\vx$ equal to $0$, we see that the inner minimization problem is optimized at $2 r_i \cdot x_i = 2 (\AA^\top \vphi)_i$ for all $i$, and equivalently $x_i = (\AA^\top \vphi)_i / r_i$. Plugging this back into the maximization objective w.r.t. $\vphi$ we obtain
\begin{align*}
\mathcal{E}_{\vr}(\vb) 
&=
 \max_{\phi} \left\langle \vr, \left(\mdiag(\vr)^{-1} \AA^\top \vphi \right)^2 \right\rangle + 2 \left \langle  \vphi, \vb - \AA \mdiag(\vr)^{-1} \AA^\top \vphi  \right\rangle 
\\
&=
 \max_{\vphi} 2 \langle \vphi, \vb \rangle - \langle \vphi, \AA\mdiag(\vr)^{-1} \AA^\top \vphi  \rangle
 \\
 &= \vb^{\top} \left( \AA \mdiag(\vr)^{-1} \AA^{\top} \right)^{+} \vb\,,
\end{align*}
where for the last equality we used that by optimality conditions one must have $(\AA\mdiag(\vr)^{-1}\AA^\top) \vphi = \vb$.

Finally, we prove (\ref{eq:energy_inverse_char}) by using the fact that for any symmetric matrix $\LL$ and vector $\vb$ one has that
\[
\frac{1}{\max_\phi 2 \vb^\top \vphi - \vphi^\top \LL \vphi} = \min_{\vphi: \vb^\top \vphi = 1} \vphi^\top \LL \vphi\,,
\]
which can be seen by observing that both expressions are optimized at $\vphi = \LL^+ \vb$, then applying (\ref{eq:energy_max_char}).

\end{proof}

\subsection{Lower Bound Lemma}
\begin{lemma}\label{lem:prod_lb} Let a set of nonnegative reals $\beta_1, \dots, \beta_k$ such that $1+\epsilon \leq \beta_i \leq \rho$ for all $i$, and $\sum_{i=1}^k \beta_i \geq S$. Then, for any $k$, one has that
\[
\prod_{i=1}^k \beta_i \geq \min\{  \rho^{S/{\rho}}, (1+\epsilon)^{S/{(1+\epsilon)}} \}\,.
\]
\end{lemma}
\begin{proof}
Consider a fixed $k$, and let us attempt to minimize the product of $\beta_i$'s subject to the constraints. Equivalently we want to minimize $\sum_{i=1}^k \log(\beta_i)$, which is a concave function. Therefore its minimizer is attained on the boundary of the feasible domain.
This means that for some $0\leq k' \leq k-1$, there are $k'$ elements equal to $1+\epsilon$, $k-1-k'$ equal to $\rho$, and one which is exactly equal to the remaining budget, i.e. $S - k'(1+\epsilon) - (k-1-k') \rho$, which yields the product $(1+\epsilon)^{k'} \rho^{k-k'-1} (S-k'(1+\epsilon)-(k-1-k')\rho)$. This can be relaxed by allowing $k$ and $k'$ to be non-integral. Hence we aim to minimize the product $(1+\epsilon)^{k'} \rho^{k-k'}$, subject to $(1+\epsilon)k' + \rho (k-k') = S$.
Finally, we observe that we can always obtain a better solution by placing all the available mass on a single one of the factors, i.e. we lower bound either by $(1+\epsilon)^{S/(1+\epsilon)}$, or $\rho^{S/\rho}$, whichever is lowest.
\end{proof}

\bibliographystyle{amsalpha}
\bibliography{ref}

\end{document}